\documentclass
{llncs}  

\newif\ifieee
\ieeefalse		

\newif\ifllncs
\llncstrue		

\newif\ifamsart
\amsartfalse             

\newif\ifproofs
\proofstrue             

\newif\ifappendix
\appendixfalse          

\newif\ifabbrev
\abbrevfalse             

\newif\ifpubrel
\pubrelfalse             

\normalfont
\usepackage[T1]{fontenc}
\usepackage{fancyhdr}

\ifllncs\usepackage{smallsubsub}\fi

\usepackage{latexsym}
\usepackage{amsmath}
\usepackage{amsfonts} 
\usepackage{amssymb}

\ifieee

\renewcommand\paragraph[1]{\par\bigskip\noindent\textbf{#1}\hspace{1ex}}
\fi

\usepackage{xy} 
\xyoption{all} 
\usepackage{wrapfig}

\newcounter{running}[section]
{\setcounter{running}{\value{enumi}}\end{enumerate}}

\def\squareforqed{\hbox{\rlap{$\sqcap$}$\sqcup$}}

\ifllncs
\spnewtheorem{eg}[theorem]{Example}{\bfseries}{\rmfamily}
\spnewtheorem{rmk}[theorem]{Remark}{\bfseries}{\rmfamily}
\else
\theoremstyle{definition} 
\newtheorem{theorem}{Theorem}
\newtheorem{definition}[theorem]{Definition}
\newtheorem{lemma}[theorem]{Lemma}
\newtheorem{corollary}[theorem]{Corollary}

\fi

\newif\ifcites
\citesfalse   

\newcommand{\newop}[2]{\newcommand{#1}{\ensuremath{\mathsf{#2}}}}
\newop{\restr}{Rstr}
\newop{\exec}{Exc}
\newop{\ES}{ES}
\newop{\entry}{entry}
\newop{\exit}{exit}
\newop{\states}{states}
\newop{\initial}{initial}
\newop{\labels}{labels}
\newop{\accept}{recv}
\newop{\deliver}{xmit}
\newop{\mpair}{mpair}
\newop{\acceptep}{acc\_ep}
\newop{\deliverep}{del\_ep}
\newop{\createch}{create}
\newop{\pends}{ends}
\newop{\ipends}{init\_ends}
\newop{\LEFT}{left}
\newop{\RIGHT}{right}
\newop{\MID}{mid}

\newop{\local}{local}
\newop{\Tag}{Tag}
\newop{\sender}{sender}
\newop{\recipt}{rcpt} 
\newop{\chan}{chan} 
\newop{\chans}{chans} 
\newop{\msg}{msg} 
\newop{\fresh}{fresh} 
\newop{\events}{ev} 
\newop{\proj}{proj} 
\newop{\loc}{loc}
\newop{\Rt}{Rt}
\newop{\Rg}{Rg}
\newop{\gr}{gr}
\newop{\un}{un}
\newop{\run}{run}
\newop{\runs}{runs}
\newop{\intf}{intf}
\newop{\Host}{Host}
\newop{\traces}{traces}
\newop{\dom}{dom}
\newop{\invisible}{invis}
\newop{\visible}{vis}
\newop{\inp}{input}
\newop{\ND}{ND}
\newop{\NI}{NI}
\newop{\lab}{lab}
\newop{\select}{select}

\newop{\lsrc}{src}
\newop{\lobs}{obs}
\newop{\lcut}{cut}
\newop{\inb}{inb}
\newop{\outb}{outb}
\newop{\Dir}{Dir}

\newop{\vis}{\mathsf{vis}}
\newop{\hid}{\mathsf{hid}}

\usepackage{enumerate}
\usepackage[]{xy}

\newcommand{\bmeas}[1]{\ensuremath{\fn{\mathsf{ms}}(\rtm,#1)}}
\newcommand{\rmeas}[2]{\ensuremath{\fn{\mathsf{ms}}(#1,#2)}}

\newcommand{\corr}[1]{\ensuremath{\fn{\mathsf{cor}}(#1)}}
\newcommand{\rep}[1]{\ensuremath{\fn{\mathsf{rep}}(#1)}}
\newcommand{\ext}[3]{\ensuremath{\fn{\mathsf{ext}}(#1,#2,#3)}}
\newcommand{\qt}[2]{\ensuremath{\fn{\mathsf{qt}}(#1,#2)}}

\newcommand{\fn}[1]{\ensuremath{\operatorname{\mathit{#1}}}}
\newcommand{\rtm}{\ensuremath{\fn{\mathsf{rtm}}}}
\newcommand{\admits}[1]{\ensuremath{\mathcal{E}(#1)}}

\newcommand{\sig}[2]{[\![\,#1\,]\!]_{#2}}
\mathchardef\mhyphen="2D
\newcommand{\attstart}{\ensuremath{\mathsf{att\mhyphen start}}}
\newcommand{\downset}[1]{\ensuremath{#1\!\!\downarrow}}
\newcommand{\upset}[1]{\ensuremath{#1\!\!\uparrow}}
\newcommand{\support}[1]{\ensuremath{\fn{Supp}(#1)}}
\newcommand{\msys}{\ensuremath{\mathcal{MS}}}
\newcommand{\asys}{\ensuremath{\mathcal{AS}}}
\newcommand{\sq}{\ensuremath{\mathsf{seq}}}
\newcommand{\corrupted}{\ensuremath{\mathsf{c}}}
\newcommand{\regular}{\ensuremath{\mathsf{r}}}
\newcommand{\mval}{\ensuremath{\mathcal{MV}}}
\newcommand{\qts}{\ensuremath{\mathcal{Q}}}
\newcommand{\rst}{\ensuremath{\mathsf{rst}}}

\newtheorem{assumption}{Assumption}

\fancypagestyle{draft}{%
  \fancyhf{}%
  \fancyhead[C]{DRAFT as of \today}%
  \fancyfoot[C]{DRAFT as of \today}%
}

\fancypagestyle{title}{%
  \fancyhf{}%
  \fancyhead[L]{Approved for Public Release; Distribution Unlimited. Case Number 15-3963.}%
  \fancyfoot[C]{This technical data was produced for the
    U.S. Government\\ under Contract. No. W15P7T-13-C-A802, and is\\
    subject to the Rights in Technical Data-Noncommercial Items clause\\
    at DFARS 252.227-7013 (FEB 2012)\\
\copyright 2016 The MITRE Corporation. ALL RIGHTS RESERVED.}%
}

\title{Principles of Layered Attestation}
\author{Paul D. Rowe\\ prowe@mitre.org}
\institute{The MITRE Corporation}

\begin{document}
\maketitle
\ifpubrel
\thispagestyle{title}
\fi

\begin{abstract}
  Systems designed with measurement and attestation in mind are often
  layered, with the lower layers measuring the layers above
  them. Attestations of such systems, which we call \emph{layered
    attestations}, must bundle together the results of a diverse set
  of application-specific measurements of various parts of the system.
  Some methods of layered attestation are more trustworthy than
  others, so it is important for system designers to understand the
  trust consequences of different system configurations. This paper
  presents a formal framework for reasoning about layered
  attestations, and provides generic reusable principles for achieving
  trustworthy results.
\end{abstract}

\section{Introduction}
\label{sec:intro}
Security decisions often rely on trust. Many computing architectures
have been designed to help establish the trustworthiness of a system
through remote attestation. They gather evidence of the integrity of a
target system and report it to a remote party who appraises the
evidence as part of a security decision. A simple example is a network
gateway that requests evidence that a target system has recently run
antivirus software before granting it access to a network. If the
virus scan indicates a potential infection, or does not offer recent
evidence, the gateway might decide to deny access, or perhaps divert
the system to a remediation network. Of course the antivirus software
itself is part of the target system, and the gateway may require
integrity evidence for the antivirus for its own security
decision. This leads to the design of layered systems in which deeper
layers are responsible for generating integrity evidence of the layers
above them.

A simple example of a layered system is one that supports ``trusted
boot'' in which a chain of boot-time integrity evidence is generated
for a trusted computing base that supports the upper layers of the
system. A more complex example might be a virtualized cloud
architecture. The virtual machines (VMs) at the top are supported at a
lower layer by a hypervisor or virtual machine monitor. Such an
architecture may be augmented with additional VMs at an intermediate
layer that are responsible for measuring the main VMs to generate
integrity evidence. These designs offer exciting possibilities for
remote attestation. They allow for specialization and diversity of the
components involved, tailoring the capabilities of measurers to their
targets of measurement, and composing them in novel ways. 

However, the resulting layered attestations are typically more complex
and challenging to analyze. Given a target system, what set of
evidence should an appraiser request? What extra guarantees are
provided if it receives integrity evidence of the measurers
themselves? Does the order in which the measurements are taken matter?
Can the appraiser tell if the correct sequence of measurements was
taken?

This paper begins to tame the complexity surrounding attestations of
these layered systems. We provide a formal model of layered
measurement and attestation systems that abstracts away the underlying
details of the measurements and focuses on the causal relationships
among component corruption, measurement, and reporting. The model
allows us to provide and justify generic, reusable strategies both for
measuring system components and reporting the resulting integrity
evidence. 

\paragraph{Limitations of measurement.}
Our starting point for this paper is the recognition of the fact that
measurement cannot \emph{prevent} corruption; at best, measurement
only \emph{detects} corruption. In particular, the runtime corruption
of a component can occur even if it is launched in a known good
state. An appraiser must therefore always be wary of the gap between
the time a component is measured and the time at which a trust
decision is made. If the gap is large then so is the risk of a
time-of-check-to-time-of-use (TOCTOU) attack in which an adversary
corrupts a component during the critical time window to undermine the
trust decision. A successful measurement strategy will limit the risk
of TOCTOU attacks by ensuring the time between a measurement and a
security decision is sufficiently small. The appraiser can then
conclude that if the measured component is currently corrupted, it
must be because the adversary performed a \emph{recent} attack. 

Shortening the time between measurement and security decision,
however, is effective only if the measurement component can be
trusted. By corrupting the measurer, an adversary can lie about the
results of measurement making a corrupted target component appear to
be in a good state. This affords the adversary a much larger window of
opportunity to corrupt the target. The corruption no longer has to
take place in the small window between measurement and security
decision because the target can already be corrupted at the time of
(purported) measurement. However, in a typical layered system design,
deeper components such as a measurer have greater protections making
it harder for an adversary to corrupt them. This suggests that to
escape the burden performing a recent corruption, an adversary should
have to pay the price of corrupting a \emph{deep} component.

\paragraph{Formal model of measurement and attestation.}
With this in mind, our first main contribution is a formal model
designed to aid in reasoning about what an adversary must do in order
to defeat a measurement and attestation strategy. Rather than forbid
the adversary from performing TOCTOU attacks in small windows or from
corrupting deep components, we consider an attestation to be
successful if the only way for the adversary to defeat its goals is to
perform such difficult tasks. Thus our model accounts for the
possibility that an adversary might corrupt (and repair) arbitrary
system components at any time. The model also features a true
concurrency execution semantics which allows us to reason more
directly about the causal effects of corruptions on the outcomes of
measurement without having to reason about unnecessary interleavings
of events. It has an added benefit of admitting a natural, graphical
representation that helps an analyst quickly understand the causal
relationships between events of an execution.

We demonstrate the utility of this formal model by validating the
effectiveness of two strategies, one for the order in which to take
measurements, the other for how to report the results in quotes from
Trusted Platform Modules (TPMs). TPM is not the only technology
available that provides a hardware root of trust for reporting. Indeed
solutions may be conceived that use other external hardware security
modules or emerging hardware support for trusted execution
environments such as Intel's SGX. However, most of the research on
attestation is based on using a TPM as the hardware root of trust for
reporting, and in this work, we follow that trend.  We formally prove
that under some assumptions about measurement and the behavior of
uncorrupted components, in order for the adversary to defeat an
attestation, he must perform some corruption which is ``difficult.''
The result is relatively concrete advice that can be applied by those
building and configuring attestation systems. By implementing our
general strategies and assumptions, layered systems can engage in more
trustworthy attestations than might otherwise result.

\paragraph{Strategy for measurement.}
An intuition manifest in much of the literature on measurement and
attestation is that trust in a system should be based on a bottom-up
chain of measurements starting with a hardware root of trust for
measurement. This is the core idea behind trusted boot processes, in
which one component in the boot sequence measures the next component
before launching it. Theorem~\ref{thm:recent or deep}, which we refer
to as the ``recent or deep'' theorem, validates this common intuition
and solidifies exactly how an adversary can defeat such bottom-up
measurement strategies. It roughly says the following:
\begin{quote} 
  If a system has measured deeper components before more shallow ones,
  then the only way for the adversary to corrupt a component $t$
  without detection is either by \emph{recently} corrupting one of
  $t$'s dependencies, or else by corrupting a component even
  \emph{deeper} in the system.
\end{quote}

\paragraph{Strategy for bundling evidence.}
Given the importance of the order of measurement, it is also important
for an attestation to reliably convey not only the outcome of
measurements, but the order in which they were taken. This point is
frequently overlooked in the literature on TPM-based
attestation. Unfortunately, the structure of TPM quotes does not
always reflect this ordering information, especially if some of the
components depositing measurement values might be dynamically
corrupted. We thus propose a particular strategy for creating a bundle
of evidence in TPM quotes designed to give evidence that measurements
were indeed taken bottom up. We show in Theorem~\ref{thm:joint
  strategy} that, under certain assumptions about the uncorrupted
measurers in the system, this strategy preserves the guarantees of
bottom-up measurement in the following sense:
\begin{quote}
  If the system satisfies certain assumptions, and the TPM quote
  formed according to our bundling strategy indicates no corruptions,
  then either the measurement were really taken bottom-up, or the
  adversary \emph{recently} corrupted one of $t$'s dependencies, or
  else the adversary corrupted an even \emph{deeper} component.
\end{quote}
Thus, any attempt the adversary makes to avoid the conditions for the
hypothesis of Theorem~\ref{thm:recent or deep} force him to validate
its conclusion nonetheless.

\paragraph{Paper structure.}
\ifabbrev
\else%
The rest of the paper is structured as
follows. Section~\ref{sec:related work} puts this paper in the context
of related research from the literature. We motivate our intuitions
and informally introduce our model in Section~\ref{sec:examples}. We
formalize these intuitions with definitions in
Section~\ref{sec:meas:defs:results}, and also apply the formalism to
justify the intuition that it is better to measure ``bottom-up.'' In
Section~\ref{sec:bundling examples}, we discuss the basics of TPMs and
provide examples of how TPMs can be misused, not providing the
guarantees one might expect. We extend our model with more definitions
in Section~\ref{sec:bundling defs} and in Section~\ref{sec:bundling}
we demonstrate an effective strategy for using TPMs to bundle
evidence. We conclude in Section~\ref{sec:conclusion} pointing to
directions for future work.  \fi

\section{Related work.}\label{sec:related work}
There has been much research into measurement and attestation. While a
complete survey is infeasible for this paper, we mention the most
relevant highlights in order to describe how the present work fits
into the larger context of research in this area. We divide the work
into several broad categories. Although the boundaries between
the categories can be quite blurry, we believe it helps to structure
the various approaches.

\paragraph{Measurement techniques.}
Much of the early work was focused on techniques for measuring
low-level components that make up a trusted computing base
(TCB). These ideas have matured into implementations such as Trusted
Boot~\cite{TBoot}. Recognizing that many security failures
cannot be traced back to the TCB, Sailer et al.~\cite{SailerZJD04}
proposed an integrity measurement architecture (IMA) in which each
application is measured (by hashing its code) before it is
launched. More recently, there has been work trying to identify and
measure dynamic properties of system components in order to create a
more comprehensive picture of the runtime state of a
system~\cite{LoscoccoWPM07,KilSANZ09,DaviSW09,WeiPRRZ10}. All these
efforts try to establish what evidence is useful for inferring system
state relevant to security decisions. The present work takes for
granted that such special purpose measurements can be taken and that
they will accurately reflect the system state. Rather, our focus is on
developing principles for how to combine a variety of these measurers
in a layered attestation. We envision a system designer choosing the
measurement capabilities that best suit her needs and using our work
to ensure an appraiser can trust the integrity of the result.

\paragraph{Modular attestation frameworks.}
Cabuk and others~\cite{CabukCPR09} have proposed an architecture
designed to support layered platforms with hierarchical
dependencies. Their design introduces trusted software into the TCB as
a software-based root of trust for measurement (SRTM). Although they
explain how measurements by the SRTM integrate with the chain of
measurements stored in a TPM, they do not study the effect corruptions
of various components have on the outcome of
attestations. In~\cite{CokerGLHMORSSS11}, Coker et al. identify five
guiding principles for designing an architecture to support remote
attestation. They also describe the design of a (layered) virtualized
system based on these principles, although there does not appear to be
a publicly available implementation at the time of writing. Of
particular interest is a section that describes a component
responsible for managing attestations. The emphasis is on the
mechanics of selecting measurement agents by matching the evidence
they can generate to the evidence requested by an appraiser. There is
no discussion or advice regarding the relative order of measurements
or the creation of an evidence bundle to reflect the order. More
recently, modular attestation frameworks
instantiating~\cite{CokerGLHMORSSS11}'s principles have been
implemented~\cite{TNC,SAMSON,OAT}. These are integrated frameworks
that offer plug-and-play capabilities for measurement and attestation
for specific usage scenarios. It is precisely these types of systems
(in implementation or design) to which our analysis techniques would
be most useful. We have not been able to find a discussion of the
potential pitfalls of misconfiguring these complex systems. Our work
should be able to help guide the configuration of such systems and
analyze particular attestation scenarios for each architecture.

\paragraph{Attestation Protocols.}
Finally, works such
as~\cite{CokerGLHMORSSS11,DelauneKRS11,DattaFGK09,RamsdellDGR14} study
the properties of attestation protocols, typically protocols that use
a TPM to report on integrity evidence provided by measurement
agents. They tend to focus on the cryptographic protections required
to secure the evidence as it is sent over a
network. \cite{CokerGLHMORSSS11} proposes a protocol that binds the
evidence to a session key, so that an appraiser can be guaranteed that
subsequent communications will occur with the appraised system, and
not a corrupted substitute. \cite{DelauneKRS11} and
\cite{RamsdellDGR14} examine the ways in which cryptographic
protections for network events interact with the long-term state of a
TPM. None of these consider the measurement activities on the target
platform itself and how corruptions of components can affect the
outcome of the protocol. In~\cite{DattaFGK09}, Datta et al. introduce
a formalism that accounts for actions local to the target machine as
well as network events such as sending and receiving
messages. Although they give a very careful treatment of the effect of
a corrupted component on an attestation, their work differs in two key
ways. First, the formalism represents many low-level details making
their proof rather complex, sometimes obscuring the underlying
principles. Second, their framework only accounts for static
corruptions, while ours is specifically designed around the
possibility of dynamic corruption and repair of system components.


\section{Motivating Examples of Measurement}
\label{sec:examples}

Consider an enterprise that would like to ensure that systems
connecting to its network provide a fresh system scan by the most
up-to-date virus checker. The network gateway should ask systems to
perform a system scan on demand when they attempt to connect. We may
suppose the systems all have some component $A_1$ that is capable of
accurately reporting the running version of the virus checker. Because
this enterprise values high assurance, the systems also come equipped
with another component $A_2$ capable of measuring the runtime state of
the kernel. This is designed to detect any rootkits that might try to
undermine the virus checker's system scan. We may assume that $A_1$
and $A_2$ are both measured by a root of trust for measurement
($\rtm$) as part of a secure boot process.

We are thus interested in a system consisting of the following
components: $\{\fn{sys},\fn{vc},\fn{ker},A_1,A_2,\rtm\}$, where
$\fn{sys}$ represents the collective parts of the system scanned by
the virus checker $\fn{vc}$, and $\fn{ker}$ represents the
kernel. Based on the scenario described above, we may be interested in
the following set of measurement events
$$ \{\rmeas{\rtm}{A_1},\rmeas{\rtm}{A_2}, \rmeas{A_1}{\fn{vc}},
\rmeas{A_2}{\fn{ker}}, \rmeas{\fn{vc}}{\fn{sys}}\} $$
where $\rmeas{o_1}{o_2}$ represents the measurement of $o_2$ by
$o_1$. These measurement events generate the raw evidence that the
network gateway can use to make a determination as to whether or not
to admit the system to the network.

If any of the measurements indicate a problem, such as a failed system
scan, then the gateway has good reason to believe it should deny the
system access to the network. But what if all the evidence it receives
looks good? How confident can the gateway be that the version and
signature files are indeed up to date? The answer will depend on the
order in which the evidence was gathered. To get some intuition for
why this is the case, consider the three different specifications
pictured in Fig.~\ref{fig:specs} for how to order the
measurements. (The bullet after the first three events is inserted
only for visible legibility, to avoid crossing arrows.)
\begin{figure}
  \[
  \begin{array}{c|c}
  \xymatrix@R=1ex@C=.1em{
    \bmeas{A_1}\ar@{->}[dr] & \attstart(n)\ar@{->}[d] & \bmeas{A_2}\ar@{->}[dl]\\
    & \bullet\ar@{->}[dl]\ar@{->}[dr] &\\
    \rmeas{A_1}{\fn{vc}}\ar@{->}[dr] & &\rmeas{A_2}{\fn{ker}}\ar@{->}[dl]\\
    & \rmeas{\fn{vc}}{\fn{sys}} & \\
  } & 
  \xymatrix@R=1ex@C=.1em{
    \bmeas{A_1}\ar@{->}[dr] & \attstart(n)\ar@{->}[d] &\bmeas{A_2}\ar@{->}[dl]\\
    & \bullet\ar@{->}[dl]\ar@{->}[dr] &\\
    \rmeas{A_1}{\fn{vc}}\ar@{->}[d] & &\rmeas{A_2}{\fn{ker}}\\
    \rmeas{\fn{vc}}{\fn{sys}} & &\\
  }\\ & \\
  \mathrm{Specification~}S_1 & \mathrm{Specification~}S_2\\ & \\
  \hline
  \multicolumn{2}{c}{~} \\
  \multicolumn{2}{c}{
   \xymatrix@R=1ex@C=.1em{
    \bmeas{A_1}\ar@{->}[dr] & \attstart(n)\ar@{->}[d]&\bmeas{A_2}\ar@{->}[dl]\\
    & \bullet\ar@{->}[dl]\ar@{->}[dr] &\\
    \rmeas{A_1}{\fn{vc}} & &\rmeas{A_2}{\fn{ker}}\ar@{->}[d]\\
     & & \rmeas{\fn{vc}}{\fn{sys}}\\
  }}\\ \multicolumn{2}{c}{~} \\
  \multicolumn{2}{c}{\mathrm{Specification~}S_3}
  \end{array}\]
  \caption{Three orders for measurement}
  \label{fig:specs}
\end{figure}
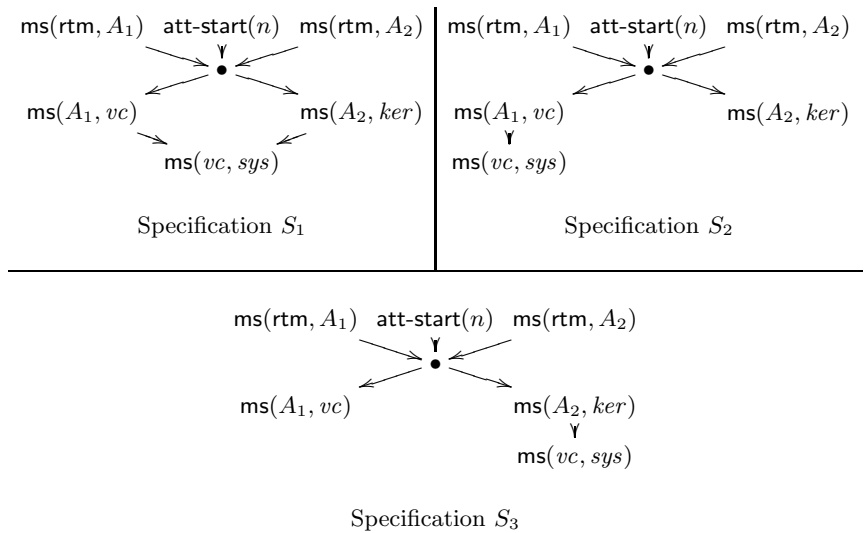

Specification~$S_1$ ensures that both $\fn{vc}$ and $\fn{ker}$ are
measured before $\fn{vc}$ runs its system scan. Specifications~$S_2$
and~$S_3$ each relax one of those ordering requirements. Let's now
consider some executions that respect the order of measurements in
each of these specifications in which the adversary manages to avoid
detection. 

Execution~$E_1$ of Fig.~\ref{fig:execs} is compatible with
Specification~$S_1$. The adversary manages to corrupt the system by
installing some user-space malware sometime in the past. If we assume
the up-to-date virus checker is capable of detecting this malware,
then the adversary must corrupt either $\fn{vc}$ or $\fn{ker}$ before
the virus scan represented by $\rmeas{\fn{vc}}{\fn{sys}}$. That is,
either a corrupted $\fn{vc}$ will lie about the results of
measurement, or else a corrupted $\fn{ker}$ can undermine the
integrity of the system scan, for example, by hiding the directory
containing the malware from $\fn{vc}$. In the case of $E_1$, the
adversary corrupts $\fn{vc}$ in order to lie about the results of the
system scan, but it does so after $\rmeas{A_1}{\fn{vc}}$ in order to
avoid detection by this measurement event.

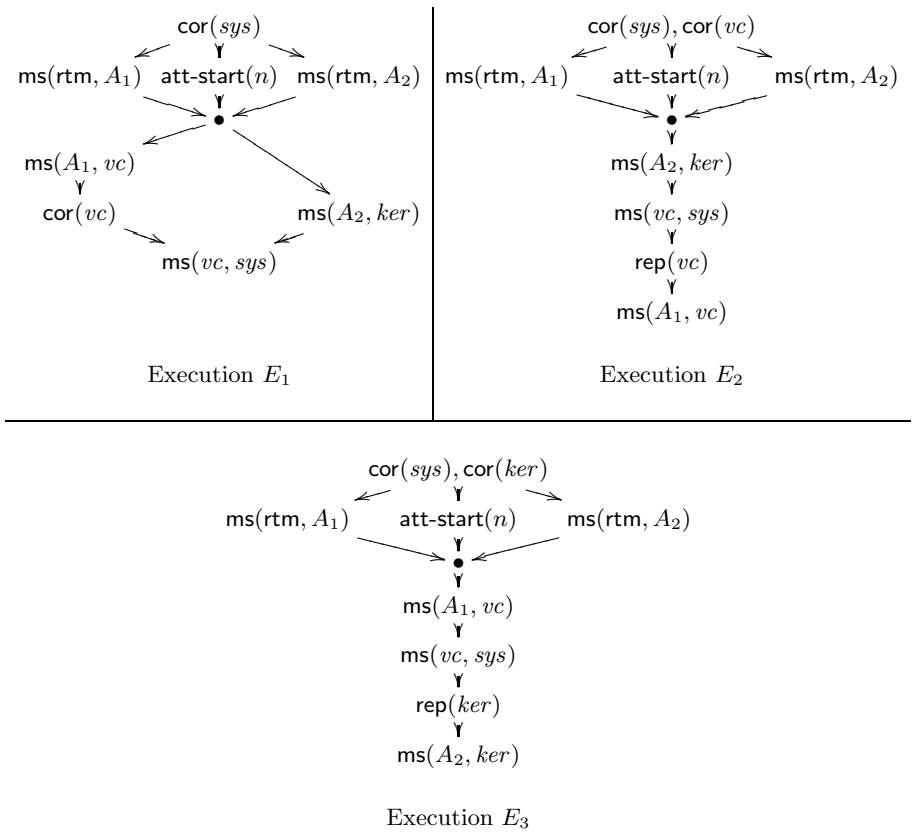
\begin{figure}
  \[
  \begin{array}{c|c}
  \xymatrix@R=1ex@C=.1em{
    & \corr{\fn{sys}}\ar@{->}[dr]\ar@{->}[dl]\ar@{->}[d] & \\
    \bmeas{A_1}\ar@{->}[dr] & \attstart(n)\ar@{->}[d] & \bmeas{A_2}\ar@{->}[dl]\\
    & \bullet\ar@{->}[dl]\ar@{->}[ddr] &\\
    \rmeas{A_1}{\fn{vc}}\ar@{->}[d]& & \\
    \corr{\fn{vc}}\ar@{->}[dr] & & \rmeas{A_2}{\fn{ker}}\ar@{->}[dl]\\
    & \rmeas{\fn{vc}}{\fn{sys}} &
  } & 
  \xymatrix@R=1ex@C=.1em{
    & \corr{\fn{sys}},\corr{\fn{vc}}\ar@{->}[dl]\ar@{->}[dr]\ar@{->}[d] & \\
    \bmeas{A_1}\ar@{->}[dr] & \attstart(n)\ar@{->}[d] & \bmeas{A_2}\ar@{->}[dl]\\
    & \bullet\ar@{->}[d] &\\
    & \rmeas{A_2}{\fn{ker}}\ar@{->}[d] &\\
    &\rmeas{\fn{vc}}{\fn{sys}}\ar@{->}[d] &\\
    & \rep{\fn{vc}}\ar@{->}[d] &\\
    & \rmeas{A_1}{\fn{vc}} &
  }\\ & \\
  \mathrm{Execution~}E_1 & \mathrm{Execution~}E_2\\ & \\
  \hline
  \multicolumn{2}{c}{~}\\
  \multicolumn{2}{c}{
  \xymatrix@R=1ex@C=.1em{
    & \corr{\fn{sys}},\corr{\fn{ker}}\ar@{->}[dl]\ar@{->}[dr]\ar@{->}[d] & \\
    \bmeas{A_1}\ar@{->}[dr] & \attstart(n)\ar@{->}[d] & \bmeas{A_2}\ar@{->}[dl]\\
    & \bullet\ar@{->}[d] &\\
    &\rmeas{A_1}{\fn{vc}}\ar@{->}[d] &\\
    &\rmeas{\fn{vc}}{\fn{sys}}\ar@{->}[d] &\\
    & \rep{\fn{ker}}\ar@{->}[d] &\\
    & \rmeas{A_2}{\fn{ker}} &
  }}\\
  \multicolumn{2}{c}{~}\\
  \multicolumn{2}{c}{\mathrm{Execution~}E_3}
  \end{array}\]
  \caption{Three system executions}
  \label{fig:execs}
\end{figure}

In Execution~$E_2$, which is consistent with Specification~$S_2$, the
adversary is capable of avoiding detection while corrupting $\fn{vc}$
much earlier. The system scan $\rmeas{\fn{vc}}{\fn{sys}}$ is again
undermined by the corrupted $\fn{vc}$. Since $\fn{vc}$ will also be
measured by $A_1$, the adversary has to restore $\fn{vc}$ to an
acceptable state before $\rmeas{A_1}{\fn{vc}}$. Execution~$E_3$ is
analagous to $E_2$, but the adversary corrupts $\fn{ker}$ instead of
$\fn{vc}$, allowing it to convince the uncorrupted $\fn{vc}$ that the
system has no malware. Since Specification~$S_3$ allows
$\rmeas{A_1}{\fn{vc}}$ to occur after the system scan, the adversary
can leverage the corrupted $\fn{vc}$ to lie about the scan results,
but must restore $\fn{vc}$ to a good state before it is measured.

Execution~$E_1$ is ostensibly harder to achieve for the adversary than
either $E_2$ or~$E_3$, because the adversary has to work quickly to
corrupt $\fn{vc}$ \emph{during} the attestation. In $E_2$ and $E_3$,
the adversary can corrupt $\fn{vc}$ and $\fn{ker}$ respectively at any
time in the past. He still must perform a quick restoration of the
corrupted component during the attestation, but there are reasons to
believe this may be easier than corrupting the component to begin
with. Is it true that all executions respecting the measurement order
of $S_1$ are harder to achieve than $E_2$ and $E_3$? What if the
adversary corrupts $\fn{vc}$ before the start of the attestation? It
would seem that he would also have to corrupt $A_1$ to avoid detection
by $A_1$'s measurement of $\fn{vc}$, $\rmeas{A_1}{\fn{vc}}$. 

One major contribution of this paper is to provide a formal framework
in which to ask and answer such questions. Within this framework we
can begin to characterize what the adversary must do in order to avoid
detection by measurement. We will show that there is a precise sense
in which Specification~$S_1$ is strictly stronger than $S_2$ or
$S_3$. This is an immediate corollary of a more general result
(Theorem~\ref{thm:recent or deep}) that validates a strong intuition
that pervades much of the literature on measurement and attestation:
Attestations are more trustworthy if the lower-level components of a
system are measured before the higher-level components. The next
section lays the groundwork for this result.

\section{Measurement Systems}
\label{sec:meas:defs:results}

\subsection{Preliminaries and Definitions}
\ifabbrev
\else
In this section we formalize the intuitions we used for the examples
in the previous section. We start by defining measurement systems
which perform the core functions of creating evidence for
attestation. 
\fi

\ifabbrev
\else
\paragraph{System architecture.}
\begin{definition}\label{def:system}
  We define a \emph{measurement system} to be a tuple
  $\msys=(O,M,C)$, where $O$ is a set of objects (e.g. software
  components) with a distinguished element $\rtm$. $M$ and $C$ are
  binary relations on $O$. We call
\begin{quote}\begin{description}
 \item[$M$] the \emph{measures} relation, and
 \item[$C$] the \emph{context} relation.
\end{description}\end{quote}
We say $M$ is \emph{rooted} when for every
  $o\in O\setminus\{\rtm\}$, $M^+(\rtm,o)$, where $M^+$ is the
  transitive closure of $M$.
\end{definition}

$M$ represents who can measure whom, so that $M(o_1,o_2)$ iff $o_1$
can measure $o_2$. $\rtm$ is the root of trust for measurement. For
this reason we henceforth always assume $M$ is rooted and $M^+$ is
acyclic (i.e. $\neg M^+(o,o)$ for any $o\in O$). This guarantees that
every object can potentially trace its measurements back to the root
of trust, and there are no measurement cycles. As a consequence,
$\rtm$ cannot be the target of measurement, i.e. for rooted, acyclic
$M$, $\neg M(o,\rtm)$ for any $o\in O$. The relation $C$ represents
the kind of dependency between $\fn{ker}$ and $\fn{vc}$ in the example
above in which one object provides a clean runtime context for
another. Thus, $C(o_1,o_2)$ iff $o_1$ contributes to maintaining a
clean runtime context for $o_2$. ($C$ stands for context.) We
henceforth always assume $C$ is transitive (i.e. if $C(o_1,o_2)$ and
$C(o_2,o_3)$ then $C(o_1,o_3)$) and acyclic. This means that no object
(transitively) relies on itself for its own clean runtime context.

Given an object $o\in O$ we define the measurers of $o$ to be
$M^{-1}(o)=\{o'\mid M(o',o)\}$. We similarly define the context for
$o$ to be $C^{-1}(o)$. We extend these definitions to sets in the
natural way.

We additionally assume $M\cup C$ is acyclic. This ensures that the
combination of the two dependency types does not allow an object to
depend on itself. Such systems are stratified, in the sense that we
can define an increasing set of dependencies as follows. 
\begin{eqnarray*}
  D^1(o) & = & M^{-1}(o)\cup C^{-1}(M^{-1}(o))\\
  D^{i+1}(o)& = & D^1(D^i(o))
\end{eqnarray*}
So $D^1(o)$ consists of the measurers of $o$ and their context. As we
will see later, $D^1(o)$ represents the set of components that must be
uncompromised in order to trust the measurement of $o$.%
\ifappendix Appendix~\ref{sec:depths} contains more examples of these
stratified dependencies $D^i$ for various measurement systems. \fi

We can represent measurement systems pictorially as a graph whose
vertices are the objects of $\msys$ and whose edges encode the $M$ and
$C$ relations. We use the convention that $M(o_1,o_2)$ is represented
by a solid arrow from $o_1$ to $o_2$, while $C(o_1,o_2)$ is
represented by a dotted arrow from $o_1$ to $o_2$. The representation
of the system described in Section~\ref{sec:examples} is shown in
Figure~\ref{fig:systems}.%
\ifappendix For examples of more measurement systems displayed in this
way, please see Appendix~\ref{sec:depths}\fi

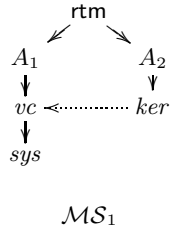
\begin{figure}
  \[
  \begin{array}{c}
   \xymatrix@R=1.5ex@C=.5em{
      & \rtm\ar@{->}[dr]\ar@{->}[dl] & \\
    A_1\ar@{->}[d] & & A_2\ar@{->}[d]\\
    \fn{vc}\ar@{->}[d] & & \fn{ker}\ar@{..>}[ll]\\
    \fn{sys} & &
   }\\ ~\\
   \msys_1
  \end{array}\]
  \caption{Visual representation of an example measurement system.}
  \label{fig:systems}
\end{figure}

\paragraph{Terms and derivability.}
It is called a measurement system because the primary activity of
these components is to measure each other. The results of measurement
are expressed using elements of a term algebra, the crucial features
of which we present next. 

Terms are constructed from some base $V$ of atomic terms using
constructors in a signature $\Sigma$. The set of terms is denoted
$\mathcal{T}_{\Sigma}(V)$. We assume $\Sigma$ includes at least some
basic constructors such as pairing $(\cdot,\cdot)$, signing
$\sig{(\cdot)}{(\cdot)}$, and hashing $\#(\cdot)$. The set $V$ is
partitioned into public atoms $\mathcal{P}$, random nonces
$\mathcal{N}$, and private keys $\mathcal{K}$. 

Our analysis will sometimes depend on what terms an adversary can
derive (or construct). We say that term $t$ is derivable from a set of
term $T\subseteq V$ iff $t\in\mathcal{T}_{\Sigma}(T)$, and we write
$T\vdash t$. We assume the adversary knows all the public atoms
$\mathcal{P}$, and so can derive any term in
$\mathcal{T}_{\Sigma}(\mathcal{P})$ at any time. For each $o\in O$, we
assume there is a distinguished set of (public) measurement values
$\mval(o)\subset\mathcal{P}$.

\paragraph{Events, outputs, and executions.}
The components $o\in O$ and the adversary on this system perform
actions. In particular, objects can measure each other and the
adversary can corrupt and repair components in an attempt to influence
the outcome of future measurement actions. Additionally, an appraiser
has the ability to inject a random nonce $n\in\mathcal{N}$ into an
attestation in order to control the recency of events.

\begin{definition}[Events]
  Let $\msys$ be a target system. An event for $\msys$ is a node~$e$
  labeled by one of the following.
\begin{enumerate}[a.]
\item A \emph{measurement event} is labeled by $\rmeas{o_2}{o_1}$ such
  that $M(o_2,o_1)$. We say such an event \emph{measures} $o_1$, and
  we call $o_1$ the \emph{target} of~$e$. We let $\support{e}$ denote
  the set $\{o_2\}\cup C^{-1}(o_2)$.
\item An \emph{adversary event} is labeled by either $\corr{o}$ or
  $\rep{o}$ for $o\in O\setminus\{\rtm\}$.
\item The \emph{attestation start} event is labeled by
  \emph{$\attstart(n)$}, where $n$ is a term.
\end{enumerate}
When an event $e$ is labeled by $\ell$ we will write $e=\ell$. We will
often refer to the label $\ell$ as an event when no confusion will
arise. 

An event $e$ \emph{touches} $o$, iff either
  \begin{enumerate}[i.]
  \item $o$ is an argument to the label of $e$, or
  \item $o\in\support{e}$.
  \end{enumerate}
\end{definition}


The $\attstart(n)$ event will serve to bound events in time. It
represents the random choice by the appraiser of the value $n$. The
appraiser will know that anything occurring after this event can
reasonably be said to occur ``recently''. Regarding the measurement
events, the $\rtm$ is typically responsible for measuring components
at boot-time. All other measurements are load-time or runtime
measurements of one component in $O$ by another. Adversary events
represent the corruption ($\corr{\cdot}$) and repair ($\rep{\cdot}$)
of components. Notice that we have excluded $\rtm$ from corruption and
repair events. This is not because we assume the $\rtm$ to be immune
from corruption, but rather because all the trust in the system relies
on the $\rtm$: Since it roots all measurements, if it is corrupted,
none of the measurements of other components can be trusted.

As we saw in the motivational examples, an execution can be described
as a partially ordered set (poset) of these events. We choose a
partially ordered set rather than a totally ordered set because the
latter unnecessarily obscures the difference between \emph{causal}
orderings and \emph{coincidental} orderings. However, due to the
causal relationships between components, we must slightly restrict our
partially ordered sets in order to make sense of the effect that
corruption and repair events have on measurement events. To that end,
we next introduce a sensible restriction to these partial orders.

A poset is a pair $(E,\prec)$, where $E$ is any set and $\prec$ is a
transitive, acyclic relation on $E$. When no confusion arises, we
often refer to $(E,\prec)$ by its underlying set $E$ and use $\prec_E$
for its order relation. Given a poset $(E,\prec)$, let
$\downset{e}=\{e'\mid e'\prec e\}$, and $\upset{e} = \{e'\mid e\prec
e'\}$. Given a set of events $E$, we denote the set of adversary
events of $E$ by $\fn{adv}(E)$ and the set of measurement events by
$\fn{meas}(E)$.

Let $(E,\prec)$ be a partially ordered set of events for
$\msys=(O,M,C)$ and let $(E_o,\prec_o)$ be the substructure consisting
of all and only events that touch $o$. We say $(E,\prec)$ is
\emph{adversary-ordered} iff for every $o\in O$, $(E_o,\prec_o)$ has
the property that if $e$ and $e'$ are incomparable events, then
neither $e$ nor $e'$ are adversary events.


\begin{lemma}\label{lem:corruption state}
  Let $(E,\prec)$ be a finite, adversary-ordered poset for $\msys$,
  and let $(E_o,\prec_o)$ be its restriction to some $o\in O$. Then
  for any non-adversarial event $e\in E_o$, the set
  $\fn{adv}(\downset{e})$ (taken in $E_o$) is either empty or has a
  unique maximal element.
\end{lemma}

\begin{proof}
  Since $(E,\prec)$ is adversary-ordered, $\fn{adv}(E_o)$ is
  partitioned by $\fn{adv}(\downset{e})$ and
  $\fn{adv}(\upset{e})$. Suppose $\downset{e}$ is not empty. Then
  since $E_o$ is finite, it has at least one maximal element. Suppose
  $e'$ and $e''$ are distinct maximal elements. Thus they must be
  $\prec_o$-incomparable. However, since $(E,\prec)$ is
  adversary-ordered, either $e'\prec_o e''$ or $e''\prec_o e'$,
  yielding a contradiction. \qed
\end{proof}

\begin{definition}[Corruption state]
  Let $(E,\prec)$ be a finite, adversary-ordered poset for
  $\msys$. For each event $e\in E$ and each object $o$ the
  \emph{corruption state} of $o$ at $e$, written $\fn{cs}(e,o)$, is an
  element of $\{\bot,\regular,\corrupted\}$ and is defined as
  follows. $\fn{cs}(e,o) = \bot$ iff $e\not\in E_o$. Otherwise, we
  define $\fn{cs}(e,o)$ inductively: %
 \[\fn{cs}(e,o) =
 \left\{ \begin{array}{cl}
     \corrupted & : e=\corr{o}\\
     \regular & : e=\rep{o}\\
     \regular & : e\in\fn{meas}(E)\wedge\fn{adv}(\downset{e})\cap E_o=\emptyset\\
     \fn{cs}(e',o) & : e\in\fn{meas}(E)\wedge
     e'\mathrm{~maximal~in~}\fn{adv}(\downset{e})\cap E_o
\end{array}
\right.\] When $\fn{cs}(e,o)$ takes the value $\corrupted$ we say $o$
is \emph{corrupt} at $e$; when it takes the value $\regular$ we say
$o$ is \emph{uncorrupt} or \emph{regular} at $e$; and when it takes
the value $\bot$ we say the corruption state is \emph{undefined}.
\end{definition}

We assume measurement events produce evidence of the corruption state
of the component. The question of measurement is tricky though,
because what counts as evidence of corruption for one appraiser might
pass as evidence of regularity by another. It is the job of
measurement to produce evidence not to evaluate it. Furthermore,
evidence of regularity (or corruption) might take many forms. In our
analysis we bracket most of these questions by making a simplifying
assumption about measurements. In particular, we assume a given
appraiser can accurately determine the corruption state of a target
given that the measurement was taken by a regular component with a
regular context. More formally, we assume the following.

\begin{assumption}[Measurement Accuracy]\label{def:outputs}
  Let $\mathcal{G}(o)$ and $\mathcal{B}(o)$ be a partition for
  $\mval(o)$. Let $e=\rmeas{o_2}{o_1}$. The \emph{output} of $e$,
  written $\fn{out}(e)$, is defined as follows. $\fn{out}(e)= v\in
  \mathcal{B}(o_1)$ iff $\fn{cs}(e,o_1)=\corrupted$ and for every
  $o\in \{o_2\}\cup\{o'\mid C(o',o_2)\}$,
  $\fn{cs}(e,o)=\regular$. Otherwise
  $\fn{out}(e)=v\in\mathcal{G}(o_1)$.

  If $\fn{out}(e) \in\mathcal{B}(o_1)$ we say $e$ \emph{detects a
    corruption}. If $\fn{out}(e) \in\mathcal{G}(o_1)$ but
  $\fn{cs}(e,o_1) = \corrupted$, we say the adversary \emph{avoids
    detection at $e$}.

  If $e=\attstart(n)$, then $\fn{out}(e) = n$.
\end{assumption}

Thus, the appraiser partitions the possible measurement values of $o$
into those that she believes indicate regularity ($\mathcal{G}(o)$)
and those that indicate corruption ($\mathcal{B}(o)$).  The output of
a measurement by regular components is in $\mathcal{G}(o)$ if $o$ is
regular at the measurement event, and in $\mathcal{B}(o)$ if it is
corrupt. We view this assumption as allowing us to explore the best
one can hope for with measurement. Of course, in reality, things are
not so rosy. Simple measurement schemes like hashing the code can
cause components to look corrupt when, in fact, a small change that is
irrelevant to security has changed the outcome of the
hash. Conversely, a runtime measurement scheme that only looks at a
subset of the component's data structures may fail to detect a
corruption and report a measurement value that looks acceptable.  One
could imagine relaxing this assumption by accounting for probabilities
of detection depending on which components have been corrupted. We
leave such investigations for future work with the understanding that
the results in this paper represent, in a sense, the strongest
conclusions one can expect from any measurement system.

We can now define what it means to be an execution of a measurement
system. 

\begin{definition}[Executions, Specifications]
  Let $\msys$ be a measurement system.
\begin{enumerate}
\item An \emph{execution} of $\msys$ is any finite, adversary-ordered
  poset $E$ for $\msys$.

\item A \emph{specification} for $\msys$ is any execution that
  contains no adversary events.
\end{enumerate}

Specification $S$ \emph{admits} an execution $E$ iff there is an
injective, label-preserving map of partial orders $\alpha:S\to E$. The
set of all executions admitted by $S$ is denoted $\admits{S}$.
\end{definition}

Measurement specifications are the way an appraiser might ask for
measurements to be taken in a particular order. The set $\admits{S}$
is just the set of executions in which the given events have occurred
in the desired order. The appraiser can thus analyze $\admits{S}$ in
advance to determine what an adversary has to do to avoid detection,
given that the events in $S$ were performed as specified.

The question of how an appraiser learns whether or not the actual
execution performed is in $\admits{S}$ is an important one. The second
half of the paper is dedicated to that problem. For now, we consider
what an appraiser can infer about an execution $E$ given that
$E\in\admits{S}$. 
\fi



\subsection{A Strategy for Measurement}
\label{sec:analyses}

We now turn to a formalization of the rule of thumb at the end of
Section~\ref{sec:examples}. By ensuring that specifications have
certain structural aspects, we can conclude the executions they admit
satisfy useful constraints. In particular, it is useful to measure
components from the bottom up with respect to the dependencies of the
system. That is, if whenever $o_1$ depends on $o_2$ we measure $o_2$
before measuring $o_1$, then we can usefully narrow the range of actions
the adversary must take in order to avoid detection. For this
discussion we fix a target system $\msys$. Recall that $D^1(o)$
represents the measurers of $o$ and their runtime context.

\begin{definition}\label{def:bottom up}
  A measurement event $e=\rmeas{o_2}{o_1}$ in execution~$E$ is
  \emph{well-supported} iff either
  \begin{enumerate}[i.]
  \item $o_2= \rtm$, or
  \item\label{cond:supp} for every $o\in D^1(o_1)$, there is a
    measurement event $e'\prec_E e$ such that $o$ is the target
    of~$e'$.
  \end{enumerate}
  When $e$ is well-supported, we call the set of $e'$ from
  Condition~\ref{cond:supp} above the
  \emph{support} of $e$.  An execution~$E$ \emph{measures bottom-up}
  iff each measurement event $e\in E$ is well-supported.
\end{definition}

\begin{theorem}[Recent or deep]\label{thm:recent or deep}
  Let $E$ be an execution with well-supported measurement event
  $e=\rmeas{o_1}{o_t}$ where $o_1\ne \rtm$. Suppose that $E$ detects no
  corruptions. If the adversary avoids detection at $e$, then either
  \begin{enumerate}
  \item there exist $o\in D^1(o_t)$ and $o'\in M^{-1}(o)$ such that
    $\rmeas{o'}{o} \prec_E \corr{o} \prec_E e$
  \item there exists $o\in D^2(o_t)$ such that $\corr{o}\prec_E e$.
  \end{enumerate}
\end{theorem}

\ifabbrev
\else
\begin{proof}
  Since the adversary avoids detection at $e$, $o_t$ is corrupt at
  $e$, and there is some $o\in \{o_1\}\cup C^{-1}(o_1)\subseteq
  D^1(o_t)$ that is also corrupt at $e$. Also, since $e$~is
  well-supported, and $o_1\ne\rtm$, we know there exists
  $e'=\rmeas{o'}{o}$ with $e'\prec_E e$. We now take cases on
  $\fn{cs}(e',o)$.

  If $\fn{cs}(e',o)=\regular$ then there must be a corruption
  $\corr{o}$ between $e'$ and $e$ satisfying Clause 1 to change its
  corruption state from $\regular$ to $\corrupted$.

  If $\fn{cs}(e',o)=\corrupted$, then since $E$ detects no
  corruptions, there must be some $o^*\in \{o'\}\cup C^{-1}(o')
  \subseteq D^2(o_t)$ such that $\fn{cs}(e',o^*)=\corrupted$. Thus
  there must be a previous corruption $\corr{o^*}\prec_E e'\prec_E e$
  satisfying Clause~2. \qed
\end{proof}
\fi

This theorem says, roughly, that if measurements indicate things are
good when they are not, then there must either be a recent corruption
or a deep corruption. This tag line of ``recent or deep'' is
particularly apt if the system dependencies also reflect the relative
difficulty for an adversary to corrupt them. By ordering the
measurements so that more robust ones are measured first, it means
that for an adversary to avoid detection for an easy compromise, he
must have compromised a measurer since it itself was measured, or
else, he must have previously (though not necessarily recently)
compromised a more robust component. In this way, the measurement of a
component can raise the bar for the adversary. If, for example, a
measurer sits in a privileged location outside of some VM containing a
target, it means that the adversary would also have to break out of
the target VM and compromise the measurer to avoid detection. The
skills and time necessary to perform such an attack are much greater
than simply compromising the end target.

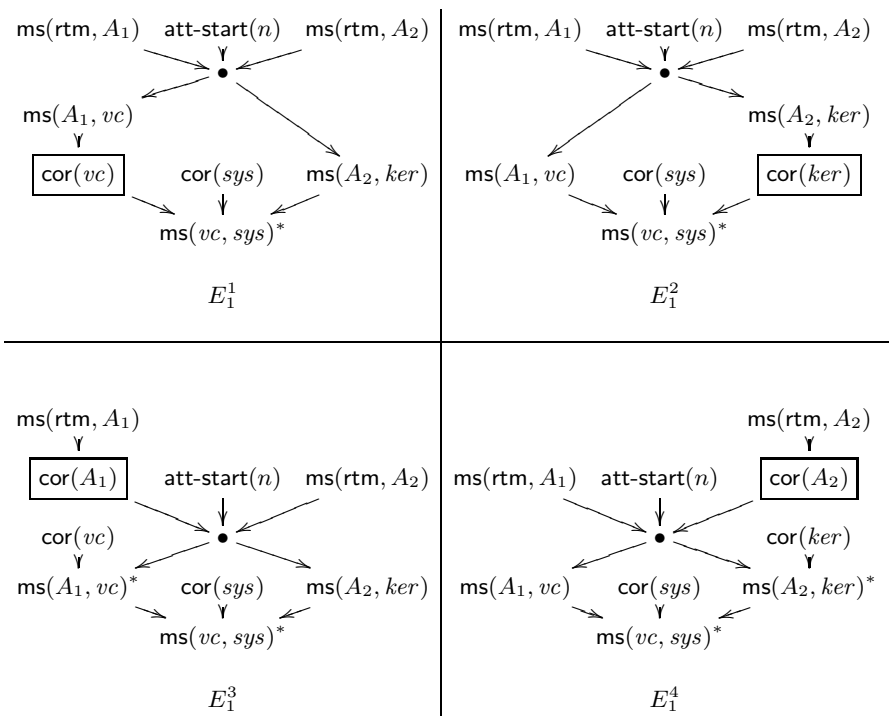
\begin{figure}
  \[
  \begin{array}{c|c}
  \xymatrix@R=1ex@C=.1em{
    \bmeas{A_1}\ar@{->}[dr] & \attstart(n)\ar@{->}[d] & \bmeas{A_2}\ar@{->}[dl]\\
    & \bullet\ar@{->}[ddr]\ar@{->}[dl] &\\
    \rmeas{A_1}{\fn{vc}}\ar@{->}[d]& & \\
    \fbox{\corr{\fn{vc}}}\ar@{->}[dr]& \corr{\fn{sys}}\ar@{->}[d]& \rmeas{A_2}{\fn{ker}}\ar@{->}[dl]\\
    & \rmeas{\fn{vc}}{\fn{sys}}^* &
  }& 
  \xymatrix@R=1ex@C=.1em{
    \bmeas{A_1}\ar@{->}[dr] & \attstart(n)\ar@{->}[d] & \bmeas{A_2}\ar@{->}[dl]\\
    & \bullet\ar@{->}[ddl]\ar@{->}[dr] &\\
    & & \rmeas{A_2}{\fn{ker}}\ar@{->}[d]\\
    \rmeas{A_1}{\fn{vc}}\ar@{->}[dr] & \corr{\fn{sys}}\ar@{->}[d] & \fbox{\corr{\fn{ker}}}\ar@{->}[dl]\\
    & \rmeas{\fn{vc}}{\fn{sys}}^* &
  }\\ & \\
  E_1^1 & E_1^2\\ & \\
  \hline
  \\ & \\
  \xymatrix@R=1ex@C=.1em{
    \bmeas{A_1}\ar@{->}[d] & & \\
    \fbox{\corr{A_1}}\ar@{->}[dr] & \attstart(n)\ar@{->}[d] & \bmeas{A_2}\ar@{->}[dl] \\
    \corr{\fn{vc}}\ar@{->}[d] & \bullet\ar@{->}[dl]\ar@{->}[dr] &\\
    \rmeas{A_1}{\fn{vc}}^*\ar@{->}[dr] & \corr{\fn{sys}}\ar@{->}[d]&\rmeas{A_2}{\fn{ker}}\ar@{->}[dl]\\
    & \rmeas{\fn{vc}}{\fn{sys}}^* &
  } &
  \xymatrix@R=1ex@C=.1em{
     & & \bmeas{A_2}\ar@{->}[d]\\
    \bmeas{A_1}\ar@{->}[dr]& \attstart(n)\ar@{->}[d] & \fbox{\corr{A_2}}\ar@{->}[dl]\\
    & \bullet\ar@{->}[dl]\ar@{->}[dr] & \corr{\fn{ker}}\ar@{->}[d]\\
    \rmeas{A_1}{\fn{vc}}\ar@{->}[dr] & \corr{\fn{sys}}\ar@{->}[d]&\rmeas{A_2}{\fn{ker}}^*\ar@{->}[dl]\\
    & \rmeas{\fn{vc}}{\fn{sys}}^* &
  }\\ & \\
  E_1^3 & E_1^4
 \end{array}\]
  \caption{Executions that do not detect corruption of $\fn{sys}$.}
  \label{fig:outcomes}
\end{figure}

Let's illustrate this result in the context of the example of
Section~\ref{sec:examples}. The specification $S_1$ satisfies the main
hypothesis of Theorem~\ref{thm:recent or deep}. Execution $E_1$
illustrates an example of the first clause of the conclusion being
satisfied. There is a ``recent'' corruption of $\fn{vc}$ in the sense
that $\fn{vc}$ is corrupted after it is measured. Since the
measurement of $\fn{vc}$ occurs after the start of the attestation,
this is truly recent, in that the adversary has very little time to
work. The appraiser can control this by ensuring that attestations
time out after some fixed amount of time.

Theorem~\ref{thm:recent or deep} also indicates other possible
executions in which the adversary can undetectably corrupt
$\fn{sys}$. There could be a recent corruption of $\fn{vc}$, or else
there could be some previous corruption of either $A_1$ or $A_2$. All
the various options are shown in Figure~\ref{fig:outcomes} in which
the measurement events at which the adversary avoids detection are
marked with an asterisk, and the corruption events guaranteed by the
theorem are boxed. Our theorem allows us to know that these
executions essentially characterize all the cases in which a corrupted
$\fn{sys}$ goes undetected.




\section{Motivating Examples of Bundling}
\label{sec:bundling examples}

The previous section discusses how to constrain adversary behavior
using the order of measurements. However, implicit in the analysis is
the assumption that an appraiser is able to verify the order and
outcome of the measurement events. Since a remote appraiser cannot
directly observe the target system, this assumption must be discharged
in some way. A measurement system must be augmented with the ability
to record and report the outcome and order of measurement events. We
refer to these additional activities as \emph{bundling} evidence. Our
focus for this paper is on using the Trusted Platform Module for this
purpose. While there are techniques and technologies that can be used
as roots of trust for reporting (e.g. hardware-based trusted execution
environments such as Intel's SGX) there has been a lot of research
into TPMs and their use for attestation. Much of that work does not
pay close attention to the importance of faithfully reporting the
order in which measurements have taken place. Thus, we believe that
studying TPM-based attestation is a fruitful place to start, and we
leave investigations of other techniques and technologies for future
work. 

\ifabbrev
\else%
\subsection{TPM Background}
Trusted Platform Modules (TPMs) are small hardware processors that are
designed to provide secure crypto processing and limited storage of
information isolated from software. Its technical specification was
written by the Trusted Computing Group (TCG)~\cite{TPM}. While TPMs
have many features designed to support subtle properties, we only
briefly review those features relevant for our purposes.

TPMs have a bank of isolated storage locations called Platform
Configuration Registers (PCRs) designed to store measurements of a
platform's state. These PCRs have a very limited interface. They start
in some well-known state and each PCR can only be updated by
\emph{extending} a new value $v$ which has the effect of updating the
contents of the PCR to be the hash of $v$ with the previous
contents. Thus the contents of each PCR serve as a historical record
of all the measurements extended into them since the most recent
system boot.

TPMs also have the ability to securely report the values in their PCRs
by creating a digital signature over their contents using a private
key that is only accessible inside the TPM. This operation is known as
a \emph{quote}. Since the PCRs are isolated from software, any remote
party that has access to the corresponding public key can verify the
contents of the PCRs. In order to protect against replay attacks and
ensure the recency of the information, TPM quotes also sign some
externally provided data, typically a random nonce chosen by an
appraiser.

Finally, TPMs have a limited form of access control for their PCRs
known as \emph{locality}. Some PCRs may only be extended by particular
privileged components. Thus if a PCR with access control enabled
contains some sequence of measurements, it must have been (one of) the
privileged component(s) that extended those values. Currently TPMs
have five localities so that they can differentiate between five
groups of components. 

Currently TPMs are widely available in commodity computers although
the surrounding architectures are such that they are rarely easy to
access and use. There has been some research into ``virtualizing''
TPMs. This entails providing robust protections for a software TPM
emulator that ensure it can achieve comparable levels of isolation
among other properties. Such a technology would be particularly useful
in virtualized cloud environments where one would like to provide the
benefits of a TPM to virtual machines that may be instantiated on
different physical hardware. Virtual TPMs (vTPMs) are currently
unavailable, however the TCG is currently producing a specification
that details the necessary protections, and there are some preliminary
implementations that will likely be modified as the details of the
specification become more clear. 

vTPMs provide two additional benefits over hardware TPMs (assuming the
necessary protections are guaranteed) that we will take advantage of
here. While hardware TPMs typically only have 24 PCRs, there is
essentially no limit on the number of PCRs a vTPM might
have. Furthermore, vTPMs would be able to implement many more than
five localities. These two features combine to allow many components
to each have dedicated access to their own PCRs. As we will see, this
is advantageous. However, given the current state of the technology,
assuming these features exist is ``forward thinking.'' The distinction
between hardware TPMs and vTPMs will not affect the core of our
analysis, so we henceforth use TPM without specifying if it is a
hardware TPM or vTPM.

\paragraph{PCR values and quotes.}\fi%
We represent both the values stored in PCRs and the quotes as terms in
$\mathcal{T}_{\Sigma}(V)$. Since PCRs can only be updated by extending
new values, their contents form a hash chain
$\#(v_n,\#(...,\#(v_1,\rst)))$. We abbreviate such a hash chain as
$\sq(v_1,\dots,v_n)$. So for example, $\sq(v_1,v_2) =
\#(v_2,\#(v_1,\rst))$. We say a hash chain $\sq(v_1,\dots,v_n)$
\emph{contains} $v_i$ for each $i\leq n$. Thus the contents of a PCR
contain exactly those values that have been extended into it. We also
say $v_i$ is \emph{contained before} $v_j$ in $\sq(v_1,\dots,v_n)$
when $i<j\leq n$. That is, $v_i$ is contained before $v_j$ in the
contents of $p$ exactly when $v_i$ was extended before $v_j$.

A quote from TPM $t$ is a term of the form
$\sig{n,(p_i)_{i\in I},(v_i)_{i\in I}}{\fn{sk}(t)}$. It is a
signature over a nonce $n$, a list of PCRs $(p_i)_{i\in I}$ and their
respective contents $(v_i)_{i\in I}$ using $\fn{sk}(t)$, the secret
key of $t$. We always assume $\fn{sk}(t)\in \mathcal{K}$ the set of
non-public, atomic keys. That means the adversary does not know
$\fn{sk}(t)$ and hence cannot forge quotes.

\subsection{Pitfalls of TPM-Based Bundling.}
The two key features of TPMs (protected storage and secure reporting)
allow components to store the results of their measurements and later
report the results to a remote appraiser. The resulting quote (or set
of quotes) is a bundle of evidence that the appraiser must use to
evaluate the state of the target system. Indeed, this bundle is the
only evidence the appraiser receives. In the rest of this section we
present various examples that demonstrate how the structure of this
bundle affects the trust inferences a remote appraiser is justified in
making about the target. 

Consider $\msys_1$ found in Section~\ref{sec:examples}, and pictured
in Fig.~\ref{fig:systems}. Ideally a remote appraiser would be able to
verify that an execution that produces a particular set of quotes
$\qts$ is in $\admits{S_1}$ (from Fig.~\ref{fig:specs}). The appraiser
must be able to do this on the basis of $\qts$ only. The possibilities
for $\qts$ depend somewhat on how $\msys_1$ is divided. For example,
if $\msys_1$ is a virtualized system, $\rtm$ might sit in an
administrative VM, and $A_1$ and $A_2$ could be in a privileged
``helper'' VM separated from the main VM that hosts
$\fn{ker},\fn{vc},$ and $\fn{sys}$. If each of these VMs is supported
by its own TPM, then $\qts$ would have to contain at least three
quotes just to convey the raw measurement evidence. However, if
$\msys_1$ is not virtualized, they might all share the same TPM and a
single quote might suffice. For our purposes it suffices to consider a
simple architecture in which all the components share a single TPM.

\paragraph{Strategy 1: A single hash chain.}
Since PCRs contain an ordered history of the extended values, the
first natural idea is for all the components to share a PCR $p$, each
extending their measurements into $p$. The intuition is that the
contents of $p$ should represent the order in which the measurements
occurred on the system. To make this more concrete, assume the
measurement events of~$S_1$ have the following output:
$\fn{out}(\rmeas{\rtm}{A_1})=v_1, \fn{out}(\rmeas{\rtm}{A_2})=v_2,
\fn{out}(\rmeas{A_1}{\fn{vc}})=v_3,
\fn{out}(\rmeas{A_2}{\fn{ker}})=v_4,
\fn{out}(\rmeas{\fn{vc}}{\fn{ker}})=v_5$. Then this strategy would
produce a single quote $Q=
\sig{n,p,\sq(v_1,v_2,v_3,v_4,v_5)}{\fn{sk}(\fn{t})}$. To satisfy the
order of $S_1$, any linearization of the measurements would do, so the
appraiser should also be willing to accept  $Q'=
\sig{n,p,\sq(v_2,v_1,v_3,v_4,v_5)}{\fn{sk}(\fn{t})}$ in which $v_1$
and $v_2$ were generated in the reverse order.


Figure~\ref{fig:shared pcr} depicts an execution that produces the
expected quote~$Q$, but does not satisfy the desired order. Since all
the measurement components have access to the same PCR, if any of
those components is corrupted, it can extend values to make it look as
though other measurements were taken although they were not. This is
particularly troublesome when a relatively exposed component like
$\fn{vc}$ can impersonate the lower-level components that measure
it. 

\begin{figure}
\[\begin{array}{c}
  \xymatrix@R=1ex@C=.1em{
    \corr{\fn{sys}}\ar@{->}[dr] & & \corr{\fn{vc}}\ar@{->}[dl]\\
    & \attstart(n)\ar@{->}[d] & \\
    & \ext{\fn{vc}}{p}{v_1}\ar@{->}[d] &\\
    & \ext{\fn{vc}}{p}{v_2}\ar@{->}[d] &\\
    & \ext{\fn{vc}}{p}{v_3}\ar@{->}[d] &\\
    & \ext{\fn{vc}}{p}{v_4}\ar@{->}[d] &\\
    & \ext{\fn{vc}}{p}{v_5}\ar@{->}[d] &\\
    & \qt{n}{p}=Q &
  }\\~\\
\mathrm{Output~of~quote~is~}Q=
\sig{n,p,\sq(v_1,v_2,v_3,v_4,v_5)}{\fn{sk}(t)}.
\end{array}\]
  \caption{Defeating Strategy 1}
  \label{fig:shared pcr}
\end{figure}

This motivates our desire to have strict access control for PCRs. This
would allow the appraiser to correctly infer which component has
provided each piece of evidence. The locality feature of TPMs could be
used for this purpose. Given the limitations of locality in the
current technology, however, it may be necessary to introduce another
component that is responsible for disambiguating the source of each
measurement into a PCR. Such a strategy would require careful
consideration of the effect of a corruption of that component, and to
include measurement evidence that it is functioning properly. For
simplicity of our main analysis we freely take advantage of the
assumption that TPMs can provided dedicated access to one PCR per
component of the system it supports, leaving an analysis of the more
complicated architecture for a more complete treatment of the subject.

\paragraph{Strategy 2: Separate hash chains.}
A natural next attempt given this assumption would be to produce a
single quote over the set of PCRs that contain the measurement
evidence. This would produce quotes with the structure $Q=
\sig{n,(p_r,p_1,p_2,p_{vc}),(s_1,s_2,s_3,s_4)}{\fn{sk}(t)}$, in which
$s_1=\sq(v_1,v_2), s_2=\sq(v_3), s_3=\sq(v_4),
s_4=\sq(v_5)$. Figure~\ref{fig:no chain} demonstrates a failure of
this strategy. The problem, of course, is that, since the PCRs may be
extended concurrently, the relative order of events is not captured by
the structure of the quote.

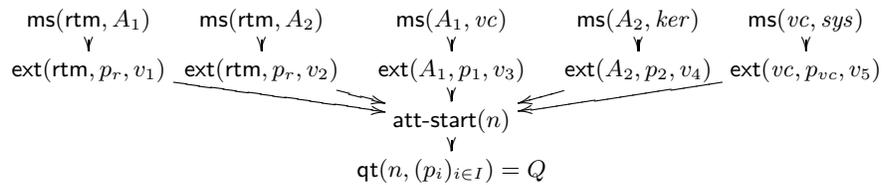
\begin{figure}
\[\begin{array}{c}
  \xymatrix@R=1ex@C=.1em{
    \rmeas{\rtm}{A_1}\ar@{->}[d] & \rmeas{\rtm}{A_2}\ar@{->}[d] &
    \rmeas{A_1}{\fn{vc}}\ar@{->}[d] & \rmeas{A_2}{\fn{ker}}\ar@{->}[d]
    & \rmeas{\fn{vc}}{\fn{sys}}\ar@{->}[d]\\
    \ext{\rtm}{p_r}{v_1}\ar@{->}[drr] &
    \ext{\rtm}{p_r}{v_2}\ar@{->}[dr] & \ext{A_1}{p_1}{v_3}\ar@{->}[d]
    & \ext{A_2}{p_2}{v_4}\ar@{->}[dl] &
    \ext{\fn{vc}}{p_{vc}}{v_5}\ar@{->}[dll]\\
    & & \attstart(n)\ar@{->}[d] & &\\
    & & \qt{n}{(p_i)_{i\in I}}=Q & &
  }\\~\\
\mathrm{Output~of~quote~is~}Q=
\sig{n,(p_r,p_1,p_2,p_{vc}),(s_1,s_2,s_3,s_4)}{\fn{sk}(t)}\\
s_1=\sq(v_1,v_2),s_2=\sq(v_3),s_3=\sq(v_4),s_4=\sq(v_5)).
\end{array}\]
  \caption{Defeating Strategy 2}
  \label{fig:no chain}
\end{figure}

\paragraph{Strategy 3: Tiered, nested quotes.}
We thus require a way to re-introduce evidence about the order of
events while maintaining the strict access control on PCRs. That is,
we should incorporate measurement evidence from lower layers before
generating the evidence for higher layers. This suggests a tiered and
nested strategy for bundling the evidence. In the case of $\msys_1$,
to demonstrate the order specified in $S_1$, our strategy might
produce a collection of quotes of the following form.%
\begin{eqnarray*}
Q_1&=&\sig{n,p_r,\sq(v_1,v_2)}{\fn{sk}(t)}\\
Q_2&=&\sig{n,(p_1,p_2),(\sq(Q_1,v_3),\sq(Q_1,v_4))}{\fn{sk}(t)}\\
Q_3&=&\sig{n,p_{vc},\sq(Q_2,v_5)}{\fn{sk}(t)}
\end{eqnarray*}
The quote $Q_1$ provides evidence that $\rtm$ has measured $A_1$ and
$A_2$. This quote is itself extended into the PCRs of $A_1$ and $A_2$
before they take their measurements and extend the results. $Q_2$
thus represents evidence that $\rtm$ took its measurements before
$A_1$ and $A_2$ took theirs. Similarly, $Q_3$ is evidence that
$\fn{vc}$ took its measurement after $A_1$ and $A_2$ took theirs since
$Q_2$ is extended into $p_{\fn{vc}}$ before the measurement evidence. 

Unfortunately, this quote structure is not quite enough to ensure that
the proper order is respected. Figure~\ref{fig:non-atomic} illustrates
the problem. In that execution, all the measurements are generated
concurrently at the beginning, and each component waits to extend the
result until it gets the quote from the layer below.  The quotes give
accurate evidence for the order in which evidence was \emph{recorded}
but not for the order in which the evidence was generated. It must be
the job of regular components to ensure that the order of extend
events accurately reflects the order of measurement events. We make
precise our assumptions for regular components in~%
\ifabbrev
\else%
Section~\ref{sec:bundling}.
\fi%
Under those extra assumptions we can prove
that a quote generated according to this final strategy is sufficient
to ensure that the execution it came from meets the guarantees of
Theorem~\ref{thm:recent or deep}. 

\begin{figure}
  \[\begin{array}{c}
  \xymatrix@R=1ex@C=.1em{
    \rmeas{\rtm}{A_1}\ar@{->}[drr] & \rmeas{\rtm}{A_2}\ar@{->}[dr] &
    \rmeas{A_1}{\fn{vc}}\ar@{->}[d] & \rmeas{A_2}{\fn{ker}}\ar@{->}[dl]
    & \rmeas{\fn{vc}}{\fn{sys}}\ar@{->}[dll]\\
    & & \attstart(n)\ar@{->}[d] & &\\
    & & \ext{\rtm}{p_r}{v_1}\ar@{->}[d] & &\\
    & & \ext{\rtm}{p_r}{v_2}\ar@{->}[d] & &\\
    & & \qt{n}{p_r}=Q_1\ar@{->}[dl]\ar@{->}[dr] & &\\
    & \ext{A_1}{p_1}{Q_1}\ar@{->}[d] & &
    \ext{A_2}{p_2}{Q_1}\ar@{->}[d] &\\
    & \ext{A_1}{p_1}{v_3}\ar@{->}[dr] & &
    \ext{A_2}{p_2}{v_4}\ar@{->}[dl] &\\
    & & \qt{n}{(p_1,p_2)}=Q_2\ar@{->}[d] & &\\
    & & \ext{\fn{vc}}{p_{vc}}{Q_2}\ar@{->}[d] & &\\
    & & \ext{\fn{vc}}{p_{vc}}{v_5}\ar@{->}[d] & &\\
    & & \qt{n}{p_{vc}}=Q_3 & &
  }\\~\\
\mathrm{Outputs~of~quotes~are~}Q_1=\sig{n,p_r,\sq(v_1,v_2)}{\fn{sk}(t)},\\
Q_2=\sig{n,(p_1,p_2),(\sq(Q_1,v_3),\sq(Q_1,v_4))}{\fn{sk}(t)},\\
Q_3=\sig{n,p_{vc},\sq(Q_2,v_5)}{\fn{sk}(t)}.
\end{array}\]
  \caption{Defeating Strategy 3}
  \label{fig:non-atomic}
\end{figure}


\ifabbrev
\else%
\section{Attestation Systems}
\label{sec:bundling defs}
\fi
\ifabbrev
\else%
In this section we augment the earlier definitions for measurement
systems to account for the use of TPMs to record and report on the
evidence generated by measurement. The following definitions closely
parallel those of Section~\ref{sec:meas:defs:results}. We begin by
expanding a measurement system into an attestation system.
\fi

\ifabbrev
\else%
\begin{definition}\label{def:system}
  We define an \emph{attestation system} to be $\asys=(O,M,C,P,L)$
  where $\msys = (O,M,C)$ is a measurement system, $P=T\times R$ for
  some set $T$ of TPMs and some index set $R$ of their PCR registers,
  and $L$ is a relation on $O \times P$.
\end{definition}

Elements of $P$ have the form $p=t.i$ for $t\in T$ and $i\in R$.  The
relation $L$ represents the access control constraints for extending
values into TPM PCRs. Each component in $O$ can only access a single
TPM, so we assume that if $L(o,t.i)$ and $L(o,t'.i')$, then $t=t'$. As
we discussed in the previous section, it is advantageous to assume the
access control mechanism dedicates a PCR to each component that needs
one. We formalize this by assuming $L$ is injective in the sense that
if $L(o,p)$ and $L(o',p)$ then $o=o'$.

The extra structure of an attestation system over a measurement system
allows us to formalize the activities of recording and reporting
evidence using events for extending values into PCRs and quoting the
results. 

\begin{definition}[Events]
  Let $\asys$ be an attestation system. An event is either an event of
  the included measurement system or it is a node labeled by one of
  the following.
\begin{enumerate}[a.]
\item An \emph{extend event} is labeled by $\ext{o}{v}{p}$, such that
  $L(o,p)$ and $v$ is a term.
\item A \emph{quote event} is labeled by $\qt{v}{t_I}$, where $v$ is a
  term, and $t_I=\{t.i\mid i\in I\}$ is a sequence of PCRs belonging
  to the same TPM $t$. We say a quote event \emph{reports on} $p$, or
  \emph{is over} $p$, if $p\in t_I$.
\end{enumerate}
The second argument to extend events and the first argument to quote
events is called the \emph{input}.

An event $e$ \emph{touches} PCR $p$, iff either
  \begin{enumerate}[i.]
  \item $e=\ext{o}{v}{p}$ for some $o$ and $v$, or
  \item $e=\qt{v}{t_I}$ for some $v$ and $p\in t_I$.
  \end{enumerate}
\end{definition}

Notice that a quote event has no corresponding component $o\in
O$. This is because TPMs may produce quotes in response to a request
by any component that has access to it. 

Just as with measurement systems, we must impose some constraints on
the partially ordered sets of these events if we expect the result of
quote and extend events to accurately represent the effects of prior
extend events. The following restriction is completely analogous to
our definition of adversary-ordered sets of events, this time focusing
on the state changes of PCRs. 

Recall that for $e\in (E,\prec)$, $\downset{e}$ is the set of events
preceding $e$ in $E$, and $\upset{e}$ is the set of events occurring
after $e$ in $E$. Let $\fn{ext}(E)$ denote the set of extend events of
$E$ and $\fn{qt}(E)$ denote the set of quote events of $E$. 

Let $(E,\prec)$ be a partially ordered set of events for
$\asys=(O,M,C,P,L)$ and let $(E_p,\prec_p)$ be the substructure
consisting of all and only events that touch PCR $p$. We say
$(E,\prec)$ is \emph{extend-ordered} iff for every $p\in P$,
$(E_p,\prec_p)$ has the property that if $e$ and $e'$ are incomparable
events, then they are both quote events.

\begin{lemma}
  Let $(E,\prec)$ be a finite extend-ordered poset for $\asys$, and
  let $(E_p,\prec_p)$ be its restriction to some $p\in P$.
  Then for every event $e\in E_p$, $\fn{ext}(\downset{e})$ is either
  empty, or it has a unique maximal event $e'$.
\end{lemma}

\begin{proof}
  Because $(E,\prec)$ is extend-ordered, $\fn{ext}(E_p)$ is
  partitioned by $\fn{ext}(\downset{e})$, $\{e\}$, and
  $\fn{ext}(\upset{e})$ for any $e\in E_p$. (The singleton $\{e\}$
  forms part of the partition exactly when $e$ is an extend event.)
  Suppose $\fn{ext}(\downset{e})$ is not empty. Since $E$ is finite,
  $\fn{ext}(\downset{e})$ has at least one maximal element. Suppose
  $e'$ and $e''$ are two distinct maximal elements. Thus they are
  $\prec_p$-incomparable. However, since $(E,\prec)$ is
  extend-ordered, either $e'\prec_p e''$ or $e''\prec_p e'$, yielding
  a contradiction. \qed
\end{proof}

This lemma allows us to unambiguously define the value in a PCR at any
event that touches the PCR. 

\begin{definition}[PCR Value]\label{def:pcr state}
  We define the \emph{value} in a PCR $p$ at event $e$ touching $p$ to
  be the following, where $\downset{e}$ is taken in $E_p$. 
\[\fn{val}(e,p) =
\left\{ \begin{array}{cl} \mathsf{rst} & :
    \fn{ext}(\downset{e})=\emptyset,
    e=\qt{n}{t_I}\\
    \#(v,\mathsf{rst}) & : \fn{ext}(\downset{e})=\emptyset,
    e=\ext{o}{v}{p}\\
    \fn{state}(e',p) & : e'=\fn{max}(\fn{ext}(\downset{e})), e=\qt{n}{t_I}\\
    \#(v,\fn{state}(e',p)) & : e'=\fn{max}(\fn{ext}(\downset{e})),
    e=\ext{o}{v}{p}
\end{array}
\right.\] 
When $e=\ext{o}{v}{p}$ we say $e$ is the event \emph{recording} the
value~$v$.
\end{definition}%

We next formalize the output of a quote event. Definition~\ref{def:pcr
  state} allows us compute all the relevant information that must be
included in a digital signature. Recall that, to ensure the signature
cannot be forged, we must assume the signing key is not available to
the adversary.

\begin{definition}[Quote Outputs]\label{def:outputs}
  Let $e=\qt{n}{t_I}$. Then its \emph{output} is
  $\fn{out}(e)=\sig{n,(t.i)_{i\in I},(v_i)_{i\in I}}{\fn{sk}(t)}$,
  where for each $i\in I$, $\fn{val}(e,t.i)=v_i$, and
  $\fn{sk}(t)\in\mathcal{K}$ (the set of atomic, non-public keys). We
  say a quote $Q$ \emph{indicates a corruption} iff some $v_i$
  contains a $v\in\mathcal{B}(o)$ for some $o$.
\end{definition}

\begin{definition}[Executions]\label{def:asys exec}
  Let $\asys$ be a target system.
  An \emph{execution} of $\asys$ is any adversary-ordered and
  extend-ordered poset $E$ for $\asys$ such that whenever $e$ has
  input $v$, then $v$ is derivable from the set
  $\mathcal{P}\cup\{\fn{out}(e')\mid e'\prec_E e\}$, i.e. the public
  terms together with the output of previous events.

  An execution $E$ \emph{produces a
    quote} $Q$ (written $E\in\admits{Q}$), iff $E$ contains a quote
  event with output $Q$.

\end{definition}
\fi


\ifabbrev
\else%
\section{Bundling Evidence for Attestation}
\label{sec:bundling}
\fi%
\ifabbrev
\else%

In this section we present several results that demonstrate some key
inferences an appraiser can make about an execution that produces a
given quote. We then formalize Strategy 3 from
Section~\ref{sec:bundling examples} for bundling evidence. Another
sequence of results demonstrates that, under certain assumptions about
the design of regular components, the guarantees of
Theorem~\ref{thm:recent or deep} are preserved for executions
producing quotes according to Strategy 3. In particular, if a
corrupted component $o$ avoids detection, then the adversary must
either have performed a recent corruption or a deep corruption
(relative to $o$).
\fi



\ifabbrev
\else%
\subsection{Principles for TPM-based bundling.}
For the remainder of this section we fix an arbitrary attestation
system $\asys=(O,M,C,P,L)$. Our first lemma allows us to infer the
existence of some extend events in an execution.

\begin{lemma}\label{lem:ext exist}
  Let $e$ be a quote event in execution $E$ with output $Q$. For each
  PCR $p$ reported on by $Q$, and for each $v$ contained in
  $\fn{val}(e,p)$ there is some extend event $e_v\prec_E e$ recording
  $v$.
\end{lemma}

\begin{proof}
  By definition, the values contained in a PCR are exactly those that
  were previously extended into it. Thus, since $\mathsf{ext}$ events
  are the only way to extend values into PCRs, there must be some
  event $e_v=\ext{o}{v}{p}$ with $e_v\prec_E e$. \qed
\end{proof}

\begin{lemma}\label{lem:inputs}
  Let $e\in E$ be an event with input parameter $v$. If
  $v\in\mathcal{N}$ or if $v$ is a signature using key
  $\fn{sk}(t)\in\mathcal{K}$, then there is a prior event $e'\prec_E
  e$ such that $\fn{out}(e')=v$.
\end{lemma}

\begin{proof}
  Definition~\ref{def:asys exec} requires $v$ to be derivable from the
  public terms $\mathcal{P}$ and the output of previous messages. Call
  those outputs $\mathcal{O}$. 

  First suppose $v\in\mathcal{N}$. Since $v$ is atomic, the only way
  to derive it is if $v\in\mathcal{P}\cup\mathcal{O}$. Since
  $\mathcal{P}\cap\mathcal{N}=\emptyset$, $v\not\in\mathcal{P}$, hence
  $v\in\mathcal{O}$ as required.

  Now suppose $v$ is a signature using key
  $\fn{sk}(t)\in\mathcal{K}$. Then $v$ can be derived in two ways. The
  first is if $v\in\mathcal{P}\cup\mathcal{O}$. In this case, since
  $v\not\in\mathcal{P}$ it must be in $\mathcal{O}$ instead as
  required. The other way to derive $v$ is to construct it from the
  key $\fn{sk}(t)$ and the signed message, say $m$. That is, we must
  first derive $\fn{sk}(t)$. Arguing as above, the only way to derive
  $\fn{sk}(t)$ is to find it in $\mathcal{O}$, but there are no events
  that output such a term. \qed
\end{proof}

\begin{lemma}\label{lem:ext order}
  Let $E$ be an execution producing quote $Q$. Assume $v_i$ is
  contained before $v_j$ in PCR $p$ reported on by $Q$, and let $e_i$
  and $e_j$ be the events recording $v_i$ and $v_j$ respectively. Then
  $e_i\prec_E e_j$.
\end{lemma}

\begin{proof}
  This is an immediate consequence of how PCR state evolves according
  to $\mathsf{ext}$ events. \qed
\end{proof}

\begin{corollary}\label{cor:quote order}
  Let $E$ be an execution producing quotes $Q$, and $Q'$ where $Q$
  reports on PCR $p$. Suppose $Q'$ is contained in $p$ before
  $v$. Then every event recording values contained in $Q'$ occurs
  before the event recording $v$.
\end{corollary}

\begin{proof}
  By Lemma~\ref{lem:ext order}, the event $e_{Q'}$ recording $Q'$ is
  before the event $e_v$ recording $v$. $Q'$ is an input to $e_{Q'}$
  satisfying the hypotheses of Lemma~\ref{lem:inputs}, hence there
  must be a prior quote event $e_q\prec_Ee_{Q'}$ with
  $\fn{out}(e_q)=Q'$. By Lemma~\ref{lem:ext exist} all events
  $e_{v_i}$ recording values $v_i$ contained in $Q'$ must occur before
  $e_q$. By the transitivity of $\prec_E$ we conclude
  $e_{v_i}\prec_Ee_v$ for each $v_i$. \qed
\end{proof}

\subsection{Formalizing and justifying a bundling strategy.}
\fi

\ifabbrev
\else%
Using these lemmas, we aim to understand the properties of an
execution $E$ if it produces a set of quotes constructed according to
Strategy 3 from Section~\ref{sec:bundling examples}. We first
formalize the tiered, nested structure of this bundling strategy.

\begin{definition}\label{def:extend bottom up}
  Let $e=\ext{o}{v}{p}$ be an extend event in execution $E$ such that
  $v\in\mval(o_t)$ for some $o_t\in O$. We say $e$ is
  \emph{well-supported} iff either
  \begin{enumerate}[i.]
  \item $o=\rtm$, or
  \item for every $o\in D^1(o_t)$ there is an extend event
    $e'\prec_Ee$ such that $e'=\ext{o'}{v'}{p'}$ with $v'\in\mval(o)$.
  \end{enumerate}
A collection of extend events $X$ \emph{extends bottom-up} iff each
$e\in X$ is well-supported.
\end{definition}

\noindent
\textbf{Bundling Strategy.}\emph{ Let $\qts$ be a set of quotes. We
  describe how to create a measurement specification $S(\qts)$. For
  each $Q\in\qts$, and each $p$ that $Q$ reports on, and each
  $v\in\mval(o_2)$ contained in $p$, $S(\qts)$ contains an event
  $e_v=\rmeas{o_1}{o_2}$ where $M(o_1,o_2)$ and $L(o_1,p)$. Similarly,
  for each $n$ in the nonce field of some $Q\in\qts$, $S(\qts)$
  contains the event $\attstart(n)$. Let $S_Q$ denote the set of
  events derived in this way from $Q\in\qts$. Then
  $e\prec_{S(\qts)}e_v$ iff $Q$ is contained before $v$ and $e\in
  S_Q$. $\qts$ complies with the bundling strategy iff $S(\qts)$
  measures bottom-up.}


\begin{proposition}\label{prop:extend bottom up}
  Suppose $E\in\admits{\qts}$ where $S(\qts)$ measures bottom-up. Then
  $E$ contains an extension substructure $X_\qts$ that extends
  bottom-up.
\end{proposition}

\begin{proof}
  Let $X_\qts$ be the subset of events of $E$ guaranteed by
  Lemma~\ref{lem:ext exist}. That is, $X_\qts$ consists of all the
  events $e=\ext{o}{v}{p}$ that record measurement values $v$ reported
  in $\qts$. For any such event $e$, if $o=\rtm$ then $e$ is
  well-supported by definition. Otherwise, since $S(\qts)$ measures
  bottom-up, Lemma~\ref{lem:ext exist} and Corollary~\ref{cor:quote
    order} ensure $X_\qts$ contain events $e'=\ext{o'}{v'}{p'}$ for
  every $o'\in D^1(o)$ where $e'\prec_E e$. Thus $e$ is also well
  supported in that case. \qed
\end{proof}

We make two key assumptions about executions of attestation
systems. 

\begin{assumption}\label{assm:prior meas}
  If $E$ contains an event $e=\ext{o}{v}{p}$ with $v\in\mval(t)$,
  where $o$ is regular at that event, then there is an event
  $e'=\rmeas{o}{t}$ such that $e'\prec_E e$. Furthermore, the most
  recent such event $e'$ satisfies $\fn{out}(e')=v$.
\end{assumption}

\begin{assumption}\label{assm:fresh meas}
  Suppose $E$ has events $e\prec_E e'$ where $e=\rmeas{o_2}{o_1}$ and
  $e'=\ext{o}{v}{p}$ where $v\in\mval(t)$, $o_1\in D^1(t)$. Then
  either 
  \begin{enumerate}
    \item $o$ is corrupt at $e'$, or
    \item there is some $e'' = \rmeas{o}{t}$ with $e\prec_E e''\prec_E
      e'$. 
  \end{enumerate}
\end{assumption}

The first assumption says that when extending measurement values
regular components only extend the value they most recently generated
through measurement. The second assumption is more complex. It is
meant to guarantee that measurements at higher layers are at least as
fresh as the measurements of the lower layers they depend on. Thus,
whenever a deeper component takes a measurement, there must be some
signal to the upper layer to tell those components to expire any
measurements they have taken. 

These two assumptions will not be validated in all attestation
systems. These are relatively subtle properties that can be expressed
in, say, SELinux policies, but would be difficult to implement in a
less constrained architecture based on a more commodity operating
system. We show that these assumptions are sufficient to ensure
Strategy~3 for bundling evidence is a good one, but they may not be
necessary. Furthermore, if a technology other than a TPM is used for
bundling, say Intel's SGX, then another set of assumptions may be more
appropriate.

\begin{theorem}\label{thm:rd-bundle}
  Let $E$ be an execution satisfying Assumptions~\ref{assm:prior meas}
  and~\ref{assm:fresh meas} that also contains an extension
  substructure $X$ that extends bottom-up. For each extend event
  $e=\ext{o_1}{v_t}{p_1}$, suppose that $v_t\in\mathcal{G}(o_t)$. Then
  for each such $e$, either
  \begin{enumerate}
  \item\label{cond:good} $e$ reflects a measurement event that is
    well-supported by measurement events reflected by the support of
    $e$.
  \item 
    \begin{enumerate}[a.]
    \item\label{cond:deep} some $o_2\in D^2(o_t)$ gets corrupted in
      $E$, or
    \item\label{cond:recent} some $o_1\in D^1(o_t)$ gets corrupted in
      $E$ after being measured.
    \end{enumerate}
  \end{enumerate}
\end{theorem}

\begin{proof}
  The proof considers an exhaustive list of cases, demonstrating that
  each one falls into one of Conditions~\ref{cond:good},
  \ref{cond:deep}, or~\ref{cond:recent}. The following diagram
  summarizes the proof by representing the case structure and
  indicating which condition each case satisfies.

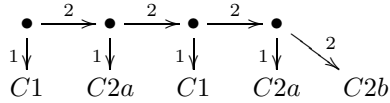
\[
    \xymatrix@R=3ex@C=1em{
    \bullet\ar@{->}_1[d]\ar@{->}^2[r]&\bullet\ar@{->}_1[d]\ar@{->}^2[r]
    & \bullet\ar@{->}_1[d]\ar@{->}^2[r] &
    \bullet\ar@{->}_1[d]\ar@{->}^2[dr] & ~\\
    C1&C2a&C1&C2a&C2b
  }
 \]

  Consider any extend event $e=\ext{o_1}{v_t}{p_1}$ of $X$ extending a
  measurement value for some $o_t\in O$. The first case distinction is
  whether or not $o_1=\rtm$. 

  \textbf{Case 1:} Assume $o_1=\rtm$. Since $\rtm$ cannot be
  corrupted, it is regular at $e$, and by Assumption~\ref{assm:prior
    meas}, $e$ reflects the measurement event $\rmeas{\rtm}{o_t}$
  which is trivially well-supported, so Condition~\ref{cond:good} is
  satisfied. 

  \textbf{Case 2:} Assume $o_1\ne\rtm$. Since $X$ extends bottom-up,
  it has events $e_i=\ext{o_2^i}{v_2^i}{p_2^i}$ extending measurement
  values $v_2^i$ for every $o^i\in D^1(o_t)$, and for each~$i$,
  $e_i\prec_E e$. Now either some $o_2^i$ is corrupt at $e_i$ (Case
  2.1), or each $o_2^i$ is regular at $e_i$ (Case 2.2).

  \textbf{Case 2.1:} Assume some $o_2^i$ is corrupt at $e_i$. Then
  there must have been a prior corruption of $o_2^i\in D^2(o_t)$, and
  hence we are in Condition~\ref{cond:deep}.

  \textbf{Case 2.2:} Assume each $o_2^i$ is regular at $e_i$. Then
  Assumption~\ref{assm:prior meas} applies to each $e_i$, so each one
  reflects a measurement event $e'_i$. In this setting, either $o_1$
  is regular at $e$ (Case 2.2.1), or $o_1$ is corrupt at $e$ (Case
  2.2.2).

  \textbf{Case 2.2.1:} Assume $o_1$ is regular at $e$. Then since the
  events $e'_i$ together with $e$ satisfy the hypothesis of
  Assumption~\ref{assm:fresh meas}, we can conclude that $e$ reflects
  a measurement event $e'=\rmeas{o_1}{o_t}$ such that $e'_i\prec_E e'$
  for each~$i$. That is, $e'$ is well-supported by the $e'_i$ events
  which are reflected by the support of $e$, putting us in
  Condition~\ref{cond:good}.

  \textbf{Case 2.2.2:} Assume $o_1$ is corrupt at $e$. Since $o_1\in
  D^1(o_t)$ one of the $e'_i$ is a measurement event of $o_1$ with
  output $v_1\in\mathcal{G}(o_1)$ since $X$ only extends measurement
  values that do not indicate corruption. Call this event $e'_*$. The
  final case distinction is whether $o_1$ is corrupt at this event
  $e'_*$ (Case 2.2.2.1) or regular at $e'_*$ (Case 2.2.2.2).

  \textbf{Case 2.2.2.1:} Assume $o_1$ is corrupt at $e'_*$. Since the
  measurement outputs a good value, some element $o_2\in
  D^1(o_1)\subseteq D^2(o_t)$ is corrupt at $e'_*$. This satisfies
  Condition~\ref{cond:deep}.

  \textbf{Case 2.2.2.2:} Assume $o_1$ is regular at $e'_*$. By the
  assumption of Case 2.2.2, $o_1$ is corrupt at $e$ with $e'_*\prec_E
  e$. Thus there must be an intervening corruption event for
  $o_1$. Since $e'_*$ is a measurement event of $o_1$, this satisfies
  Condition~\ref{cond:recent}. \qed
\end{proof}

Assumption~\ref{assm:fresh meas} can only guarantee that an object is
remeasured whenever one of its dependencies is remeasured. It cannot
ensure that all orderings of $S(\qts)$ are preserved in
$\admits{\qts}$. For this reason we introduce the notion of the core
of a bottom-up specification. The \emph{core} of a bottom-up
specification $S$ is the result of removing any orderings between
measurement events $e_i\prec_S e_j$ whenever $e_i$ is not in the
support of $e_j$. That is, the core of $S$ ignores all orderings that
do not contribute to $S$ measuring bottom-up. 

\begin{theorem}\label{thm:joint strategy}
  Let $E\in\admits{\qts}$ such that $S(\qts)$ measures bottom-up, and
  let $S'$ be its core. Suppose that $\qts$ detects no corruptions,
  and that $E$ satisfies Assumptions~\ref{assm:prior meas}
  and~\ref{assm:fresh meas}. Then one of the following holds:
  \begin{enumerate}
  \item $E\in\admits{S'}$,
  \item there is some $o_t\in O$ such that 
    \begin{enumerate}[a.]
    \item some $o_2\in D^2(o_t)$ is corrupted, or
    \item some $o_1\in D^1(o_t)$ is corrupted after being measured.
    \end{enumerate}
  \end{enumerate}
\end{theorem}

\begin{proof}
  By Proposition~\ref{prop:extend bottom up}, $E$ contains a substructure
  $X_\qts$ of extend events that extends bottom-up. Thus by
  Theorem~\ref{thm:rd-bundle}, Conditions~2a and~2b are
  possibilities. So suppose instead that $E$ satisfies
  Condition~\ref{cond:good} of Theorem~\ref{thm:rd-bundle}. We
  must show that $E\in\admits{S'}$. In particular, we construct
  $\alpha:S'\to E$ and show that it is label- and
  order-preserving. 

  Consider the measurement events $e_i^s$ of $S'$. By construction,
  each one comes from some measurement value $v_i$ contained in
  $\qts$. Similarly, the well-supported measurement events $e_i^m$ of
  $E$ guaranteed by Theorem~\ref{thm:rd-bundle} are reflected by
  extend events $e_i$ of $E$ which are, in turn, those events that
  record each $v_i$ in $\qts$. We let $\alpha(e_i^s)=e_i^m$ for
  each~$i$.

  To see that $\alpha$ is label-preserving, consider first the label
  of $e_i^s$. It corresponds to a measurement value $v_i$ contained in
  some $p_i$ of $\qts$. So $e_i^s$ is labeled $\rmeas{o}{o'}$ where
  $M(o,o')$, $v_i\in\mval(o')$, and $L(o,p_i)$. The event $e_i^m$ also
  corresponds to the same $v_i$. Lemma~\ref{lem:ext exist} ensures
  that $e_i=\ext{o}{v}{p_i}$ with $L(o,p_i)$, and so the measurement
  event it reflects is $e_i^m=\rmeas{o}{o'}$ with $M(o,o')$ and
  $v_i\in\mval(o')$. Thus $e_i^s$ and $e_i^m$ have the same label. 

  We now show that if $e_i^s\prec_{S(\qts)}e_j^s$ then $e_i^m\prec_E
  e_j^m$. The former ordering exists in $S'$ because some quote
  $Q\in\qts$ is contained in $p_j$ before $v_j$ and $v_i$ is contained
  in $Q$, and because $e_i^s$ is in the support of $e_j^s$. By
  Corollary~\ref{cor:quote order} $e_i\prec_E e_j$ and $e_i$ is in the
  support of $e_j$ and therefore Theorem~\ref{thm:rd-bundle}
  ensures that the measurements they reflect are also ordered, i.e.
  $e_i^m\prec_E e_j^m$. 

  Finally, consider any events $e=\attstart(n)$ in $S'$. They come
  from nonces $n$ found as inputs to quotes $Q\in\qts$. By
  Lemma~\ref{lem:inputs}, $E$ also has a corresponding event $e^*$
  with $\fn{out}(e^*)=n$. Since $\attstart$ events are the only ones
  with output of the right kind, $e^*=\attstart(n)$ as well. Thus we
  can extend $\alpha$ by mapping each such $e$ to the corresponding
  $e^*$. The rules for $S(\qts)$ say that $e\prec_{S(\qts)}e'$ only
  when $Q$ has $n$ in the nonce field, and $Q$ occurs before the value
  recorded by $e'$. In $E$, $e^*$ precedes the event producing $Q$ (by
  Lemma~\ref{lem:inputs}) which in turn precedes $e'$ by
  Lemmas~\ref{lem:inputs} and~\ref{lem:ext order}. Thus the orderings
  in $S(\qts)$ involving $\attstart$ events are also preserved by
  $\alpha$. \qed
\end{proof}

\fi

\section{Conclusion}
\label{sec:conclusion}
\ifabbrev
\else%
In this paper we have developed a formalism for reasoning about
layered attestations. Within this framework we have justified the
intuition (pervasive in the literature on measurement and attestation)
that it is important to measure a layered system from the bottom up
(Theorem~\ref{thm:recent or deep}). We also proposed and justified a
strategy for using TPMs to bundle evidence
(Theorem~\ref{thm:rd-bundle}). If used in conjunction, these two
results guarantee an appraiser that if an adversary has corrupted a
component and managed to avoid detection, then it must have performed
a recent or deep corruption (Theorem~\ref{thm:joint strategy}).  \fi%

Although we used our model to justify the proposed general and
reusable strategies for layered attestations, we believe our model has
a wider applicability. It admits a natural graphical interpretation
that is straightforward to understand and interpret. Future work to
develop reasoning methods within the model could lead to more
automated analysis of attestation systems. We believe a tool that
leverages automated reasoning and the graphical interpretation would
be a useful asset. 

For the present work we made several simplifying assumptions. For
instance, we assumed that if measurers (or their supporting
components) are corrupted, then they can always forge the results of
measurement. This conservative, worst-case view does not account for
a situation in which, say, even if the OS kernel is corrupted, it may
still be hard to forge the results of a virus scan. Conversely, we
also assumed that uncorrupted measurers can always detect
corruptions. This is certainly not true in most systems. Adapting the
model to account for probabilities of detection would be an
interesting line of research that would make the model applicable to a
wider class of systems.

Another issue of layered attestations that we did not address here,
is the question of what to do when the system components fall into
different administrative domains. This would be typical of a cloud
architecture in which the lower layers are administered by the cloud
provider, but the customers may provide their own set of measurement
capabilities as well. A remote appraiser must be able to negotiate an
attestation according several policies. Our model might be extended to
account for the complexities that arise. 

Finally, we chose to study the use of TPMs for bundling evidence. We
believe other approaches leveraging timing-based techniques or other
emerging technologies including hardware-supported trusted execution
environments such as Intel's new SGX instruction set could be captured
similarly. This would allow us to formally demonstrate the security
advantages of one approach over another, or understand how to build
attestation systems that leverage several technologies.

\section*{Acknowledgments}
I would like to thank Pete Loscocco for suggesting and guiding the
direction of this research. Many thanks also to Perry Alexander and
Joshua Guttman for their valuable feedback on earlier versions of this
work.  Finally, thanks also to Sarah Helble and Aaron Pendergrass for
lively discussions about measurement and attestation systems.


\bibliography{semantics}
\bibliographystyle{plain}

\end{document}